\documentclass[12pt]{article}
\usepackage[letterpaper, margin=1in, headheight=15pt]{geometry}
\usepackage[T1]{fontenc}
\usepackage{amsmath,amsfonts,amsxtra,mathrsfs,latexsym,amsthm,amssymb,amscd,mathabx,tabularx,epsfig,scrextend,epic,setspace,ulem,comment,relsize,enumerate,subfig,flexisym}
\usepackage{quoting,xparse}
\usepackage[english]{babel}
\usepackage[colorinlistoftodos]{todonotes}
\usepackage[metapost]{mfpic}
\usepackage{graphicx,natbib} 
\usepackage{blindtext}
\usepackage[colorlinks = true,
            linkcolor = blue,
            urlcolor  = blue,
            citecolor = blue,
            anchorcolor = blue]{hyperref} 
\theoremstyle{plain}
\newtheorem{theorem}{Theorem}
\newtheorem{proposition}{Proposition}

\newtheorem{corollary}{Corollary}
\newtheorem{lemma}{Lemma}

\newtheorem{step}{Step}
\newtheorem*{lemma*}{Lemma}

\theoremstyle{definition}
\newtheorem{definition}{Definition}

\newtheorem{remark}{Remark}
\newtheorem{axiom}{Axiom}

\DeclareMathOperator*{\argmax}{argmax}

\makeatletter
\newcommand{\varnotsuccsim}{\mathrel{\mathpalette\varn@t\succsim}}
\newcommand{\varn@t}[2]{%
  \vphantom{/{#2}}%
  \ooalign{\hfil$\m@th#1/\mkern2mu$\cr\hfil$\m@th#1#2$\hfil\cr}%
}
\makeatother

\title{Hoping for the best while preparing for the worst in the face of uncertainty\thanks{We are deeply grateful to Marc Fleurbaey, Antonin Macé, William Thomson, and Frank Riedel for their support and their advice. We also thank (in alphabetical order) Susumu Cato, Alain Chateauneuf, Maya Eden, Soo Hong Chew, Jose Heleno Faro, Lorenz Goette, Chen Li, Marcus Pivato, Xi Zhi ``RC'' Lim, Lasse Mononen, Xiangyu Qu and Peter Wakker for their valuable comments. This paper was presented in the theory seminar of the University of Rochester, the economic theory lunch seminar of the Bielefeld University, the TOM seminar of the Paris School of economics, the theory seminar of the Karlsruhe Technical University, the 17th SCW meeting  (Paris), the TUS conference (Paris), the RUD conference (Manchester), and the SDM conference (Chengdu), of which we thank all the participants. We warmly acknowledge the help of an anonymous referee whose detailed and numerous suggestions greatly improved the paper.}}
\author{Pierre Bardier\thanks{Paris School of Economics, \'Ecole Normale Sup\'erieure de Paris (Paris Sciences et Lettres). E-mail: pierre.bardier@ens.fr} \and Bach Dong-Xuan\thanks{Center for Intelligence Economic Science, Southwestern University of Finance and Economics. E-mail: bachdongx@gmail.com} \and Van-Quy Nguyen\thanks{Universit\'e Paris-Saclay, Univ Evry, EPEE, France. Faculty of Mathematical Economics, National Economics University, Vietnam. E-mail: quynv@neu.edu.vn}}
\date{}

\begin{document}
\maketitle
\begin{abstract}   

We propose and axiomatize a new model of incomplete preferences under uncertainty, hope-and-prepare preferences: one act is preferred to another if and only if both its optimistic evaluation ---the welfare level attained in the best-case scenario--- and its pessimistic evaluation — the welfare level attained in the worst-case scenario--- rank it above the other. Both evaluations are computed over sets of probability distributions. We make the case that, compared to existing incomplete criteria under ambiguity, hope-and-prepare preferences address the trade-off between conviction and decisiveness in a new way, which is more favorable to decisiveness. We characterize a completion of an incomplete hope-and-prepare preference relation admitting an (asymmetric) $\alpha$-\textit{maxmin expected utility} representation, in which $\alpha$ is unique. Conversely, given a standard $\alpha$-MEU preference and a fixed value of $\alpha$, we recover the concordant hope-and-prepare preference that it completes.


\end{abstract}

\section{Introduction}
\newcommand{\chapquote}[3]{\begin{quotation} \textit{#1} \end{quotation} \begin{flushright} - #2, \textit{#3}\end{flushright} }

\chapquote{``Hoping for the best, prepared for the worst, and unsurprised by anything in between.''}{Maya Angelou}{\textit{I Know Why the Caged Bird Sings}.}
 
The complexity of economic decisions is likely to result in agents' inability or unwillingness to decide over the uncertain options they are supposed to compare. In this regard, the restrictiveness of the assumption that individual preferences be complete was early acknowledged,\footnote{For instance, \cite{Au1962} wrote: ``\textit{Of all the axioms of utility theory, the completeness axiom is perhaps the most questionable. Like others of the axioms, it is inaccurate as a description of real life; but unlike them, we find it hard to accept even from a normative viewpoint.}'' \cite{S1989}, commenting on his characterization of the maxmin criterion, depicted the completeness axiom as ``\textit{the most restrictive and demanding assumption}.''} and was recently highlighted by empirical studies.\footnote{See \cite{CETTOLIN2019547}, \cite{nielsen2022revealed}.} We propose and characterize a new \textit{incomplete} decision criterion according to which, in the face of Knightian uncertainty \citep{knight1921risk}, agents \textit{both hope for the best and prepare for the worst}.

We study preferences over acts $f: S \to X$, which are mappings from states of the world to outcomes, and we introduce and axiomatize preferences $\succsim$ admitting the following representation:

\begin{align}\label{introrep}
f \succsim g \iff \begin{cases}
\min_{p\in C} \int u(f) dp \geq \min_{p\in C} \int u(g) dp\\[4pt]
\max_{p\in D} \int u(f) dp \geq \max_{p\in D} \int u(g) dp
\end{cases},
\end{align}
where $u$ is a numerical representation of preferences over outcomes, and $C$ and $D$ are sets of probability distributions over the states, interpreted as sets of different \textit{scenarios}.\footnote{The function $u: X \to \mathbb{R}$ is non-constant, affine and unique up to positive affine transformation. The sets $C$ and $D$ are unique, non-disjoint, compact and convex.} We choose to work with a preorder to properly distinguish between indifference and incomparability. A decision maker (DM) following such a criterion ranks an act $f$ above an act $g$ if and only if $f$ provides a higher expected utility than $g$ in the worst-case scenario in $C$ as well as in the best-case scenario in $D$. The idea that \textit{both} worst-case and best-case evaluations serve as reference points for decision under uncertainty, in the case of complete preferences, dates back at least to \cite{hurwicz1951optimality}, and was explored in the multiple-prior setting by \cite{marinacci2002probabilistic} and \cite{GM2004}. 

Our criterion is based on the conjunction of an optimistic (or ambiguity-seeking) assessment and of a pessimistic (or ambiguity-averse) assessment.\footnote{$C$ and $D$ being non-disjoint, the expected utility in the best-case scenario is higher.} We then interpret a DM with such a preference as hoping for the best scenario to realize, while also preparing for the worst one to happen, when evaluating each option: we thus refer to a preference relation admitting such a representation as a \textit{hope-and-prepare preference}.
 
We shall give special attention in our analysis to the \textit{concordant case} in which $C=D$ ---which we also characterize, adding two original axioms to the set of axioms characterizing the general representation (\ref{introrep}). Acts are then evaluated according to \textit{the interval} of all expected utility levels that they induce across all possible scenarios. More precisely, an act $f$ is preferred to an act $g$ if and only if \textit{any expected utility level that is attainable from $g$ but not from $f$ is below any expected utility level that is attainable from $f$, and there exists at least one level that is indeed attainable from $g$ but not from $f$.} This intuitive criterion for comparing ranges of expected utility levels corresponds to  \textit{the strong set order}, which is, arguably, the most common way to compare intervals. 

\bigskip
The conjunction of a best-case evaluation and of a worst-case evaluation at play in our criterion is akin to the one at play in the notion of \textit{obvious manipulation} \citep{TM2020}, defined for revelation games in which the uncertainty faced by an agent concerns others' messages. Accordingly, the significant practical relevance of the notion of obvious manipulation provides support for our criterion within uncertain strategic environments. This notion gives an explanation, for instance, of untruthful reporting strategies that have been consistently observed in the \textit{Immediate Acceptance mechanism}, used to match students with schools.\footnote{See \cite{pathak2008leveling} and \cite{dur2018identifying}.}

The scope for applications of our criterion goes beyond strategic interactions. The notion that extreme scenarios guide individual decisions is recognized for various social and economic domains where ambiguity is present. In this regard, let us simply mention the evaluation of financial assets (\cite{BGGZ2010}, \cite{schroder2011investment}, \cite{https://doi.org/10.3982/QE243}), or the evaluation of different medical treatments by physicians and patients (\cite{doi:10.7326/0003-4819-138-5-200303040-00028}, \cite{taylor2017framework}); we discuss a third example in more detail.

It is not unusual for practitioners, reporters or fans to evaluate ``young prospects'' participating in the annual \textit{Draft} in North-American sports leagues---we take the example of the National Basketball Association league (NBA)---according to ``ceiling and floor scenarios.''\footnote{See, for example, James Hansen, \href{https://www.slcdunk.com/nba-draft/2023/6/9/23755886/what-makes-an-nba-draft-prospect-high-ceiling-or-high-floor-2023-mock-draft-big-board?fbclid=IwY2xjawG-oldleHRuA2FlbQIxMAABHWVQi0kFSpS3SFLaRaW1a9GZNA0zi8Z6SO9_UyuVtn8AydXOzzJPTWEF2A_aem_hKrElRHNqPZanqS2qRwhlg}{``\textit{What makes an NBA Draft prospect high ceiling or high floor?''}}, SLC Dunk, June 2023, and Kyle Boone, \href{https://www.cbssports.com/nba/news/nba-draft-2024-ceiling-and-floor-scenarios-the-best-or-worst-case-projections-for-five-top-prospects/}{``\textit{NBA Draft 2024 ceiling and floor scenarios: The best or worst case projections for five top prospects''}}, CBS Sports, June 2024. We note that the use of the expressions ``ceiling'' and ``floor'' suggests that any case lying ``in between'' is considered possible.} This can be formulated in our framework. There is potentially a myriad of parameters that the agent considers relevant for the evaluation of prospects: a state is a particular configuration of parameters.\footnote{A state may thus encompass the rosters of coaches and players, at the beginning of the season and after the winter ``trade'' period, of each franchise, their financial capacities, the performance of players already in the league, the progression of each of these prospects, the approach to officiating favored by the league's executives, \textit{etc}.} In this complex environment, the agent faces ambiguity and must compare prospects on the basis of a set $C$ of probability distributions over configurations of parameters. Loosely speaking, each player is identified with an act $f$, indicating their overall performance in each state, which is then evaluated according to a utility function $u$, and, for every scenario $p \in C$, the agent can compute the expectation of $u(f)$ according to $p$. 

The incompleteness of a criterion such as ours reflects the \textit{necessity to have sufficient conviction} when declaring that a player is more promising than an other one. On the other hand, a criterion should not be \textit{too incomplete}; let us illustrate this point by comparing our criterion to two alternative ones. Given $u$ and $C$, the agent could require, for a ``player $f$'' to be declared more promising than a ``player $g$'', that, for each scenario in $C$, the expected utility associated with $g$ be lower than the expected utility level associated with $f$ (\cite{B2002}). One could even require that any expected utility level attainable from $g$ be lower than any expected utility level attainable from $f$ (\cite{E2022})---that the ``ceiling'' of $g$ be lower than the ``floor'' of $f$. Both of these conditions are stronger than condition (\ref{introrep}), expressing a more demanding notion of sufficient conviction. However, it may very well be the case that only ``generational talents'' such as Victor Wembanyama,\footnote{See, for example, Sam Harris, \href{https://www.bbc.com/sport/olympics/articles/crg7kgyg9r0o}{``\textit{Why `alien' Wembanyama is France's next big thing - literally''}}, BBC Sports, July 2024.}  (who was present in the 2023 Draft) be distinguished from other players on the basis of these more conservative criteria, and that for rather homogeneous cohorts such as the 2024 cohort, the agent fail to rank any player above an other one.\footnote{See, for example, Adam Finkelstein, \href{https://www.cbssports.com/nba/news/no-stars-have-revealed-themselves-in-the-2024-nba-draft-but-history-tells-us-theyre-hiding-in-plain-sight/.}{``\textit{No stars have revealed themselves in the 2024 NBA Draft, but history tells us they're hiding in plain sight''}}, CBS Sports, June 23 2024.} In practical terms, according to our criterion, a ``player $f$'' is declared more promising than a ``player $g$'' if and only if anything that $g$ could achieve and that $f$ could not is considered worse than anything $f$ could achieve. With hope-and-prepare preferences, which, in this case, compare players on the basis of the associated ranges of expected utility according to the strong set order, the \textit{trade-off between decisiveness and conviction} is addressed in a way that is more favorable to decisiveness. 

\bigskip
The original axiom involved in our characterization is interpreted along this line:  Axiom \ref{axiomtransiincomp_new} provides a relatively weak sufficient condition for \textit{comparability}. Our axiomatization maintains the assumption that preferences are complete over constant acts, deemed as the simplest ones. According to Axiom \ref{axiomtransiincomp_new}, if \textit{i)} the DM cannot compare $f$ to the constant $x$, while she declares $x$ more desirable than $g$ and, on the other hand, \textit{ii)} she cannot compare $g$ to the constant act $y$, while she declares $f$ more desirable than $y$, then she declares $f$ more desirable than $g$. Thus, \textit{two specific aligned pieces of evidence are enough} to conclude that an act is better than another, and Axiom \ref{axiomtransiincomp_new} may be seen as formulating a minimal departure from the completeness of a standard expected utility preference relation.

\color{black}
Furthermore, in order to account for typical situations in which agents \textit{have to} choose between two options, even if they lack conviction to express a clear preference between them in the first place, we study the completion of hope-and-prepare preferences.\footnote{From a theoretical point of view, studying a completion of an incomplete preference relation enables to use standard mathematical tools, for example for utility maximization and welfare analysis.} We demonstrate that the \textit{invariant biseparable complete extension}
of a hope-and-prepare preference admits an \textit{asymmetric\footnote{``Asymmetric'' refers to the fact that best and worst cases may be taken on different sets of scenarios.} $\alpha$-maxmin expected utility} ($\alpha$-MEU) representation---and a standard $\alpha$-maxmin representation if the hope-and-prepare preference is concordant. Notably, the asymmetric $\alpha$-MEU retains much of the tractability of the standard $\alpha$-MEU---which is beneficial for applications---while remaining flexible enough to accommodate mixed ambiguity attitudes \citep{CE2022}.\footnote{Specifically, it captures ambiguity-averse behavior for \textit{large/moderate-likelihood} events, ambiguity-seeking behavior for \textit{small-likelihood} events, and \textit{source-dependent} ambiguity attitudes \citep{CE2022}.} Importantly, in the representation we obtain, \textit{the weight $\alpha$ does not depend on the considered acts, and is unique whenever the extended hope-and-prepare preference is incomplete.} Conversely, given a standard $\alpha$-MEU preference with fixed $\alpha$, we can recover the hope-and-prepare preference whose complete extension is this $\alpha$-MEU preference relation. 

Finally, answering two natural questions of comparative statics that emerge from the proposition of a new type of incomplete preference under ambiguity, we compare the degree of incompleteness of our criterion to that of Bewley preferences
(\cite{B2002}) and of twofold preferences (\cite{E2022}), and we provide a way
to compare the ambiguity attitudes of two hope-and-prepare preferences.

\bigskip
Our paper is organized as follows. We situate our approach within the literature in Section \ref{Related lit}. We define the formal framework and introduce our criterion in Section \ref{sec:representation}. In Section \ref{sec:characterization}, we present the two main representation results. In Section \ref{sec:extension}, we investigate the completion of our criterion. Section \ref{sec:comp} is dedicated to the comparative statics questions mentioned above. In Section \ref{sec:choquet}, we characterize a class of preference relations combining a pessimistic and an optimistic Choquet integrals. The conclusions are presented in Section \ref{sec:conclusion}. All proofs can be found in the appendix.

\section{Related literature}\label{Related lit}
 
A DM hoping for the best while also preparing for the worst responds to uncertainty by combining opposite ambiguity attitudes. In this perspective, one may interpret a DM with a hope-and-prepare preference as requiring that her optimistic self and her pessimistic self be unanimous for her to rank some act above another one. The idea that the DM consists of multiple (strategic) selves appears frequently in behavioral economics.\footnote{\cite{TS1981}, \cite{BENABOU2002419}, \cite{10.1257/aer.96.5.1449}, \cite{10.1257/aer.98.4.1312}.} In recent works, \cite{CE2022} and \cite{Xia2020} provided axiomatizations for preferences involving two selves, called by the former \textit{dual-self expected utility}. Their representation differs from ours in that the agent's final decision is to be interpreted as the result of a \textit{leader-follower game} between two selves.\footnote{In order to evaluate an act, the optimistic self moves first by choosing a collection of beliefs that maximizes the evaluation of the pessimistic self, who chooses the belief that minimizes the expected utility of the act among the set chosen by the optimistic self.} In contrast, in our representation, the decision is induced by a requirement of unanimity imposed by the agent herself on the assessments of her two selves.

Our representation is connected to the concept of obvious manipulation proposed in the context of mechanism design by \cite{TM2020}. A revelation mechanism is said to be \textit{non-obviously manipulable} if, for any agent and any potential untruthful report from her, revealing her own type leads to a more desirable outcome in both of the following cases: when the others' reports are the most favourable to her, and when they are the least favourable. In our model, in the same spirit, an option---such as an untruthful report in the mechanism design context---is abandoned for an alternative only if this alternative leads to preferred outcomes in both the best and the worst scenarios among given sets of probability measures. 

The relation of our contribution to the concept of non-obvious manipulation mirrors that of \cite{E2022} to the concept of obvious dominance, due to \cite{Li2017}. Informally, in a direct mechanism, for each agent, given a belief on the profile of strategic reports of other agents (a scenario), each possible report about her own type (which corresponds to an act), induces an expected payoff. When the set of scenarios according to which all acts are evaluated is the set of distributions whose support is included in the Cartesian product of other agents' type spaces, the act $f$ is preferred to the act $g$ by a \textit{twofold multi-prior preference} if and only if the strategy corresponding to $f$ obviously dominates the strategy corresponding to $g$, and, on the other hand, $f$ is preferred to $g$ by a \textit{hope-and-prepare preference} if and only if the strategy corresponding to $f$ dominates the strategy corresponding to $g$ in the sense of \cite{TM2020}.\footnote{We refer the reader to Appendix A of \cite{bardier2024hoping}.}

Hope-and-prepare preferences define a partial order on acts. Pioneering work by \cite{Au1962}, \cite{B2002} and \cite{DFO2004} studied the representation of incomplete preferences under risk and uncertainty. Incomplete preferences in non-deterministic environments have been the object of a growing literature: see, for example, 
\cite{nascimento2011}
\cite{GK2012}, \cite{OOR2012},
\cite{FARO2015},
\cite{MS2015}, 
\cite{Hill2016}, \cite{KARNI2020}, \cite{cusumano2021} and \cite{E2022}. The closest model of incomplete preference to ours, apart from those studied in \cite{B2002} and \cite{E2022}, both compared to ours in the introduction, is introduced in \cite{nascimento2011}. As a special case of their main result, they study a criterion in which the DM considers several sets of scenarios, in each of which the performance of an act is evaluated according to the worst-case expected utility level. Then, an act is preferred to another one if and only if it performs better in each set of scenarios. Hope-and-prepare preferences enable to capture a different type of ambiguity attitude, through the consideration of the optimistic assessment. We discuss in more details how our work relates to \cite{B2002}, \cite{nascimento2011} and \cite{E2022} in the next sections.

In line with Hurwicz's approach for decision making under complete ignorance (\cite{hurwicz1951optimality}), the $\alpha$-MEU model was proposed to capture the idea that, under ambiguity, worst and best expected utility levels, over \textit{one} set of probability measures, can serve as sufficient statistics for the DM: she then computes an $\alpha$-weighted average of these levels (\cite{marinacci2002probabilistic}, \cite{kopylov2002alpha}, \cite{GM2004}). Among the recent explorations of (variants of) the $\alpha$-maxmin model,\footnote{\cite{CHATEAUNEUF2007538}, \cite{EICHBERGER20111684}, \cite{GUL2015465}, \cite{FRY2022}, \cite{klibanoff2022foundations}, \cite{HARTMANN2023105719}, \cite{hill2023beyond}, \cite{chateauneuf2024alpha}. Let us also mention the \textit{geometric} $\alpha$-MEU model (\cite{binmore}) and the ordinal Hurwicz expected utility (\cite{grant2020worst}).} the one of \cite{FRY2022} is particularly important for the way we characterize the \textit{asymmetric} $\alpha$-maxmin model as representing the completion of hope-and-prepare preferences. In the \textit{objective and subjective rationality} framework proposed by \cite{GMMS2010}, they show that the \textit{invariant biseparable complete extension} of a Bewley preference admits a standard $\alpha$-MEU representation. 

We show that the asymmetric $\alpha$-maxmin representation characterizes the invariant biseparable complete extension of a hope-and-prepare preference relation. It reduces to standard $\alpha$-maxmin in the case of a concordant hope-and-prepare preference relation. Beyond the fact that we consider the completion of a new type of preferences, our result has three salient features: the representation is asymmetric in general, $\alpha$ does not depend on the considered acts, and is unique. We also investigate how to recover a hope-and-prepare preference from a given $\alpha$-MEU preference.

\section{Setup and representation}\label{sec:representation}

We introduce the decision-theoretic framework, define (concordant) hope-and-prepare preferences, and discuss how they relate to several models of incomplete preferences under ambiguity.

\subsection{Model} 
Our analysis is conducted in the classical framework proposed by \cite{AA1963}. Uncertainty is modeled through a set $S$ of \textit{states of the world}, endowed with an algebra $\Sigma$ of subsets of $S$ called \textit{events}, and a non-empty set of consequences $X$, which is a non-singleton convex subset of a real vector space. \textit{A simple act} is defined as a function $f:S \rightarrow X$ which takes finitely many values and is measurable with respect to $\Sigma$; we denote by $\mathcal{F}$ the set of all simple acts. The \textit{mixture} of two simple acts $f$ and $g$, for any $\alpha \in [0,1]$, denoted by $\alpha f + (1-\alpha)g$, is then defined by setting, for each $s \in S$, $[\alpha f+ (1-\alpha)g](s) = \alpha f(s) + (1-\alpha) g(s)$. With the usual slight abuse of notation, for all $x\in X$, we use $x$ to denote the constant act defined by $f_x (s) = x$ for all $s\in S$. We use $\Delta$ to denote the set of all finitely additive probability distributions on $(S,\Sigma)$, endowed with the weak* topology.\footnote{The set of finitely additive bounded measures on $(S,\Sigma)$ is the dual of the set of all measurable real-valued bounded functions on $(S,\Sigma)$. Thus the weak* topology on $\Delta$ is defined according to the following convergence notion: we say that a sequence $\{p_n\}$ of elements of $\Delta$ converges to $p \in \Delta$ if for all measurable bounded function $\varphi: S \to \mathbb{R}$, $\int \varphi dp_n$ converges to $\int \varphi dp$.} We refer to a measure $p \in \Delta$ as a \textit{scenario} according to which simple acts are evaluated.\footnote{From now on, we refer to simple acts as ``acts''.}

We consider a DM whose preference is represented by a binary relation $\succsim \: \subseteq \mathcal{F} \times \mathcal{F}$. The relation $\succsim$ is a preorder, \textit{i.e.} it is reflexive and transitive. We use the standard notation $f \succsim g$ to denote $(f,g) \in \: \succsim$, and interpret $f \succsim g$ as reflecting the fact that the DM considers $f$ to be at least as desirable as $g$ \textit{with sufficient conviction}. We write $f \succ g$ if $f \succsim g$ and $g \varnotsuccsim f$, and say that $f$ is \textit{strictly preferred} to $g$. We write $f \sim g$ if $f \succsim g$ and $g \succsim f$. If $f \varnotsuccsim g$ and $g \varnotsuccsim f$, we write $f \Join g$, and say that $f$ and $g$ are \textit{incomparable}. Importantly, working with a preorder enables us to capture the difference between indifference and incomparability.

We denote the set of vectors whose $k$ elements are non-negative by $\mathbb{R}^k_+$, the set of vectors whose $k$ elements are positive by $\mathbb{R}^k_{++}$, for any natural number $k$.

\subsection{Hope-and-prepare preferences}
\subsubsection{Definition}
Our representation involves multiple priors:\footnote{\cite{EMT2012} and \cite{GM2016} both provide a review of the ways in which ambiguity, and ambiguity attitudes, have been modeled in order to offer alternatives to the traditional Bayesian framework. Multiple prior models stand out as one of the main lines of research.} the DM has a set of relevant beliefs according to which she evaluates acts. 
\begin{definition}\label{defunanimrep}
A binary relation $\succ$ is a \textit{hope-and-prepare} preference if 
\begin{align*}
f \succsim g \iff \begin{cases}
\min_{p\in C} \int u(f) dp \geq \min_{p\in C} \int u(g) dp \\[4pt]
\max_{p\in D} \int u(f) dp \geq \max_{p\in D} \int u(g) dp
\end{cases},
\end{align*}
where $u$ is a non-constant affine function defined on $X$, and $C$ and $D$ are two compact and convex subsets of $\Delta$ with $C\cap D \neq \emptyset$.

The representation is \textit{concordant} if $C=D$. 
\end{definition}

We sometimes write that $\succsim$ admits the representation $\big(u,C,D\big)$ to refer to the hope-and-prepare representation given in Definition \ref{defunanimrep}. We obtain in our axiomatization the uniqueness up to affine transformation of $u$, and the uniqueness of $C$ and $D$. \textit{We sometimes write, then, as a shortcut, that $\succsim$ admits the unique representation $(u,C,D)$.}\footnote{As opposed to writing that $\succ$ admits representation $(u,C,D)$, where $u$ is unique, \textit{up to affine transformation}, and $C$ and $D$ are unique.}

\subsubsection{Discussion}

Consider such a preference relation $\succsim$, with representation $(u,C,D)$. An act $f$ is at least as desirable as an act $g$ if and only if $f$ gives at least as high an expected utility as $g$ when they are evaluated according to their best-case scenario in $D$, \textit{and} gives at least as high an expected utility as $g$ when they are evaluated according to their worst-case scenario in $C$. This conjunction of an optimistic (or ambiguity-seeking) assessment and of a pessimistic (or ambiguity-averse) assessment models a DM hoping for the best scenario to realize, while also preparing for the worst one, when she evaluates each option.

The combination of two such opposite ambiguity attitudes may also be interpreted in the perspective of a DM consisting of two selves: for the DM to consider with sufficiently strong conviction that an act $f$ is at least as desirable as an act $g$, it is necessary, and sufficient, that her optimistic self and pessimistic self be unanimous over the ranking of $f$ and $g$.

Let $\succsim$ be a hope-and-prepare preference relation with representation $(u,C,C)$. Then, the DM evaluates any act $f$ in terms of its range $R(f)=\{\int u(f)d\mu : \mu \in C\}$ of possible expected utility levels, which, as $C$ is convex and compact, is a closed interval. Consider another act $g \in \mathcal{F}$ and suppose that $R(f)=[a,b]$ and $R(g)=[c,d]$. Then $f$ is at least as desirable as $g$ if and only if $a \geq c$ and $b \geq d$. The act $f$ is at least as desirable as $g$ if and only if any expected utility level that is attainable from $g$ but not from $f$ is below any expected utility level that is attainable from $f$, and there exists at least one level that is indeed attainable from $g$ but not from $f$. This intuitive criterion for comparing ranges of expected utility levels is \textit{the strong set order}, applied to the special case of intervals.

\paragraph{Relation to other incomplete criteria under ambiguity.} At this point, it is interesting to describe how the way ranges are compared according to concordant hope-and-prepare preferences can be formally related to the way in which they are compared according to concordant \textit{twofold preferences}, introduced by \cite{E2022}.

Let us insist on the fact that we study a preorder on $\mathcal{F}$, in contrast to \cite{E2022} who work with a strict preference relation ---for the sake of comparison, we present a natural ``preorder version'' of their representation. In our view, it is important to distinguish indifference from incomparability, if only through the fact that an act $f$ is better described as belonging to the same indifference class as itself than as being incomparable to itself. A previous version of this work was written assuming a strict preference relation, and, as can be checked in \cite{bardier2024hoping}, while proofs are lengthier, the analysis is very similar.

\begin{definition}(\cite{E2022})\label{defuntwofold}
A binary relation $\succsim$ is a (multi-prior) \textit{twofold preference} if 
$$f \succsim g \iff
\min_{p\in C} \int u(f) dp \geq \max_{p\in D} \int u(g) dp, $$
where $u$ is a non-constant affine function defined on $X$, $C$ and $D$ are two compact and convex subsets of $\Delta$ with $C\cap D \neq \emptyset$. The representation is said concordant if $C=D$.\footnote{They obtain the uniqueness, up to affine transformation, of $u$, and the uniqueness of $C$ and $D$ in their axiomatization.}
\end{definition} 

Consider $\succsim_{HP}$ a concordant hope-and-prepare preference and $\succsim_T$ a concordant twofold preference with the same representing utility function $u$ on $X$ and the same set of scenarios $C \in \Delta$. The former is always \textit{more complete} than the latter, in the sense that it is an extension of it---a general statement, beyond the case of concordant preferences, is given in Proposition \ref{prop_comp_Bew_Two}. This observation is the basis on which we compared concordant hope-and-prepare preferences to concordant twofold preferences in the NBA example of the introduction, and a similar comparison to Bewley preferences should be made (again, see Proposition \ref{prop_comp_Bew_Two}): 

\begin{definition}\label{defbewley}(\cite{B2002})
A binary relation $\succsim$ is a (multi-prior) \textit{Bewley preference} if 
$$f \succsim g \iff
 \int u(f) dp \geq  \int u(g) dp \text{ for all } p \in C,$$
where $u$ is a non-constant affine function defined on $X$, $C$ is a non-empty compact and convex subset of $\Delta$.\footnote{Similarly to the expression we used for our criterion, we will say that the twofold preference $\succsim_T$ admits the unique representation $(u,C_T,D_T)$, and that the Bewley preference $\succsim_B$ admits the unique representation $(u,C_B)$ to refer to the fact that $u$ is unique up to affine transformation, and $C_T$, $D_T$ and $C_B$ are unique.}   
\end{definition}

As we highlighted in the NBA example, with hope-and-prepare preferences, the trade-off between decisiveness and conviction is addressed in a way that is more favorable to decisiveness, compared to twofold preferences and to Bewley preferences. That is, our criterion still reflects the necessity for the DM to have sufficient conviction when declaring an act at least as desirable as another, while it induces more choices.\medskip

A special case of the preferences studied in \cite{nascimento2011} is, as ours, defined by the conjunction of different assessments.

\begin{definition}(\cite{nascimento2011})
A binary relation $\succsim$ is a \textit{N}\&\textit{R preference} if 
$$f \succsim g \iff
 \min_{p \in C}\int u(f) dp \geq  \min_{p \in C}\int u(g) dp \text{ for all } C \in \mathcal{C},$$
where $u$ is a non-constant affine function defined on $X$, $\mathcal{C}$ is a class of non-empty compact and convex subsets of $\Delta$.    
\end{definition}

Our approach and criterion differ from those of \cite{nascimento2011} in several aspects. In terms of methodology, in contrast to them, we provide an axiomatization directly on the domain of simple acts.\footnote{Their result is obtained on the set of lotteries on simple acts.} In addition, in our representation result, the pair of sets of probability measures is unique.\footnote{In their representation, $\mathcal{C}$ is not unique. The authors obtain the uniqueness of the closure of the convex hull of $\mathcal{C}$, where the set of subsets of the simplex is endowed with the Hausdorff topology.}

Furthermore, while the set $\mathcal{C}$ may be infinite in their model---the DM then takes decisions based on the unanimity of an arbitrary, potentially infinite, number of selves---our criterion requires the conjunction of merely two evaluations. 

From a behavioral perspective, N\&R preferences are based on the unanimity of a collection of MEU representations, reflecting a pessimistic attitude toward ambiguity. However, substantial experimental evidence points to more nuanced patterns of ambiguity attitudes: the same individuals may exhibit ambiguity aversion in some decision problems while displaying ambiguity-seeking behavior in others---see \cite{TK2015} for a survey. In contrast to N\&R preferences, hope-and-prepare preferences encompass, in addition to a pessimistic perspective, an optimistic one, and therefore, do not exhibit a systematic attitude toward ambiguity. In line with this feature of a hope-and-prepare preference relation, which is incomplete, we show that an invariant biseparable completion of such a preference admits an asymmetric $\alpha$-MEU representation, which can accommodate a mixture of different ambiguity attitudes (see Section \ref{sec:extension}). 

Finally, we emphasized the importance of the special case of concordant hope-and-prepare preferences, defined by a very intuitive comparison of ranges of expected utility: this has no counterpart in the N\&R model.

\medskip
\noindent
\textbf{Non-concordant hope-and-prepare preferences.} We note that in the case of a non-concordant representation, the two sets of beliefs may be obtained as the DM distorts the probability measures of a single set, possibly informed by expert sources. The DM may distort probabilities differently when adopting an optimistic versus a pessimistic perspective, as illustrated in the following example. 

\medskip

\noindent
\begin{quote}
\small
\textbf{Example.} Consider an NBA fan who may bet on the next game involving her favorite team, team A, which, sadly, happens not to be a ``500 team'' in the current season.\footnote{This expression is used to denote teams that have lost the majority of their games during the current season.} The opponent, team B, who has a better record, is more likely to win the game. The set of states of the world is the set of possible scores for the teams, $S=[70,140] \times [70,140]$, of generic element $(s_A,s_B)$. The set of events $\Sigma$ is the set of all subsets of $S$. The DM may take a bet on the subset $S_{\{A \text{ wins}\}}=\{ (s_A,s_B) \in S \mid s_A\geq s_B\}$---either A wins, or the teams have equal scores before the overtime starts---on the basis of several probability measures on $(S, \Sigma)$ reflecting the opinions of different experts or brokers. Let $C$ be (the convex hull of) this set of probability measures, and assume that the value associated with $S_{\{A \text{ wins}\}}$ across measures in $C$ browses $[\frac{1}{3},\frac{1}{2}]$. 

The DM has a utility function over money $u: \mathbb{R} \to \mathbb{R}$, $u(x)=\frac{1}{2}x$, and must compare, in particular, the two following bets, or acts: for all $s \in S$,
\begin{align*}
    &f(s)=50\$ \text{ if } s \in S_{\{A \text{ wins}\}},\\
    &f(s)=-30\$ \text{ otherwise;}\\
    &g(s)=34\$ \text{ if } s \in S_{\{A \text{ wins}\}}^c,\\
    &g(s)=-30\$ \text{ otherwise,}
\end{align*}
where $S_{\{A \text{ wins}\}}^c$ is the complement of $S_{\{A \text{ wins}\}}$. The act $g$ amounts to betting 30 dollars \textit{against} team A. 

If the DM is consequentialist in the sense that she evaluates acts exclusively according to the payoff they may induce, then $C$ appears as the natural set of scenarios to use: according to the hope-and-prepare preference relation $\succsim$, with representation $(u,C,C)$, $g\succsim f$.\footnote{The minimal and maximal expected utility level induced by $f$ is $25 \times \frac{1}{3}-15\times \frac{2}{3}\approx-1.67$, and the maximal one is $25 \times \frac{1}{2}-15\times \frac{1}{2}=5$. The minimal expected utility level induced by $g$ is $17 \times \frac{1}{2}-15\times \frac{1}{2}=1$, and the maximal one is $17 \times \frac{2}{3}-15 \times \frac{1}{3}\approx 6.33$.}

On the other hand, non-concordant hope-and-prepare preferences can capture the non-consequentialist reluctance of the DM to bet against her favorite team. This type of non-consequentialism has long been documented in the psychology literature on ``valence'', has been established in the sport context by several recent empirical studies \citep{morewedge2018betting,kossuth2020does,donkor2023identity}, and has motivated the axiomatic work of \cite{adam2024event}. In the context of this example, starting from $C$, the DM, while hoping for the best, may, for instance, distort probability measures giving too high a value to states---which she dislikes---in which team A loses by too high a margin. Let $D$ be this set of distorted probability measures, and assume that the value associated with $S_{\{A \text{ wins}\}}$ across measures in $D$ browses $[\frac{1}{3}+0.02,\frac{1}{2}+0.02]$. Assume that, on the contrary, when she performs a pessimistic assessment, the DM does not depart from the scenarios given by the experts. According to the hope-and-prepare preference relation $\succsim'$, with representation $(u,C,D)$, $f \Join' g$.\footnote{The minimal expected utility level induced by $f$ over $C$ is $25 \times \frac{1}{3}-15\times \frac{2}{3}\approx-1.67$, and the maximal one over $D$ is $25 \times (\frac{1}{2}+0.02)-15\times (\frac{1}{2}-0.02)=5.8$. The minimal expected utility level induced by $g$ over $C$ is $17 \times \frac{1}{2}-15\times \frac{1}{2}=1$, and the maximal one over $D$ is $17 \times (\frac{2}{3}-0.02)-15 \times (\frac{1}{3}+0.02)\approx 5.69$.} Thus, in contrast to the concordant case above, given these potential rewards, the DM does not prefer betting 30\$ against her favorite team compared to betting 30\$ in favor of it.\footnote{Note that the incomparability of $f$ and $g$ would not be obtained with the concordant preference defined using the convex hull of $C$ and $D$. Let $\Tilde{C}$ denote this convex hull. Then, the value associated with $S_{\{A \text{ wins}\}}$ across measures in $\tilde{C}$ browses $[\frac{1}{3},\frac{1}{2}+0.02]$. The minimal and maximal expected utility levels induced by $f$ over $\tilde C$ are $-1.67$ and $5.8$, respectively, while the minimal and maximal expected utility levels induced by $g$ are $0.36$ and $6.33$, respectively. Thus, according to the hope-and-prepare preference relation $\tilde{\succsim}$, with representation $(u,\tilde{C}, \tilde{C})$, $g \tilde{\succsim} f$.} Of course, sufficiently high monetary rewards in case team B wins can compensate the reluctance of the DM to bet against her favorite team when her preference is $\succsim'$, as can be checked by replacing $34\$$ by $42\$$ in the definition of $g$.
\end{quote}

\normalsize
\medskip
\noindent
\textbf{Best and worst outcomes in the lab.} The phenomenon that, in general, individuals overweight extreme outcomes, when they face known probabilities, as well as when they face ambiguous ones, is well documented. A large body of experimental studies---surveyed in \cite{TK2015}---have highlighted the presence of such \textit{insensitivity} among student subjects, and subsequent research, conducted among the general population, has confirmed its importance \citep{dimmock2016ambiguity,watanabe2024ambiguity}.\footnote{The reader can find a thorough account of the numerous studies pointing out insensitivity in the Online Appendix D.2 of \cite{baillon2025source}.} The idea that options should be compared \textit{only} based on extreme outcomes corresponds to complete insensitivity.\footnote{Our analysis models ambiguity rather than risk, as the DM has multiple priors. Despite this difference, the psychological mechanisms behind the two approaches are closely related, and findings obtained for outcomes are relevant for models involving multiple priors. The connection between them can be made explicit as follows---we are grateful to Peter Wakker for pointing out this connection. \textit{For a given act}, the DM can associate each prior with the induced expected utility (EU) of the act. Assume that the DM computes an \textit{EU-certainty-equivalent} (EU-CE) of each prior, given this act, that is, an outcome yielding a utility level equal to the EU level induced by the prior. Each act is now associated with a set of certainty-equivalents, one for each prior, and acts are compared on the basis of these sets. In other words, each act induces a \textit{mapping from priors to their EU-CE}. And the highest and lowest expected utility levels according to which the act is evaluated with hope-and-prepare preferences are equal to the highest and lowest values taken by the corresponding mapping.}

\section{Representation results}\label{sec:characterization}
In this section, we present the main representation results. We first axiomatize hope-and-prepare preferences involving different sets of priors. We then identify two additional axioms, original to this work, under which the two sets of priors coincide.

\subsection{Characterization of hope-and-prepare preferences}

We now proceed to an axiomatic characterization of hope-and-prepare preferences based on the following six axioms. Axioms \ref{axiomclassic},\ref{axiomcont}, \ref{axiomcertain}, and \ref{axiommono} express standard requirements, Axiom \ref{axiomconv} was proposed in \cite{E2022}, and Axiom \ref{axiomtransiincomp_new} is original to this work.

\begin{axiom}
\label{axiomclassic}
The restriction of $\succsim$ to $X$ is a non-trivial complete preorder.\footnote{That is, for all $x, y \in X$, $x \succsim y$ or $y \succsim x$; and there exist $x, y \in X$ with $x \succ y$.}
\end{axiom}

\begin{axiom}
\label{axiomcont}
For all $(f,g,h) \in \mathcal{F}^3$, the sets $\{\alpha\in [0,1]:\alpha f+(1-\alpha)g \succsim h\}$ and $\{\alpha\in [0,1]:h \succsim \alpha f+(1-\alpha)g\}$ are closed.
\end{axiom}

\begin{axiom}
\label{axiomcertain}
For all $f,g\in \mathcal{F}$, $x\in X$, and $\alpha\in (0,1)$, $f\succsim g$ if and only if $\alpha f+(1-\alpha)x \succsim \alpha g+(1-\alpha)x$.
\end{axiom}

The relation $\succsim$ is a preorder that is complete on constant acts; the asymmetric part $\succ$ and the symmetric part $\sim$ are derived in the usual way. Axiom \ref{axiomcont} is the standard Archimedean continuity condition adopted in models of decision under uncertainty. Axiom \ref{axiomcertain} is the certainty independence axiom proposed by \cite{GS1989} in their seminal paper as a weakening of the independence axiom at play in the characterization of subjective expected utility.

\begin{axiom}\label{axiomconv}
For all $x\in X$, the sets $\{f\in \mathcal{F}:f\succsim x\}$ and $\{f\in \mathcal{F}: x\succsim f\}$ are convex.
\end{axiom}

Axiom \ref{axiomconv} states that comparisons to a given constant act should not be sensitive to hedging. The convexity of $\{f\in \mathcal{F}:f\succsim x\}$ reflects uncertainty aversion: an act obtained through hedging between two acts that are at least as desirable as the constant act $x$ is also at least as desirable as $x$. The convexity of $\{f\in \mathcal{F}: x\succsim f\}$ reflects preference for uncertainty: when the DM considers two uncertain acts at most as desirable as $x$, an act obtained through hedging between the two is also at most as desirable as $x$.

We require in Axiom \ref{axiommono} that $\succsim$ be consistent with state-wise dominance:

\begin{axiom}
\label{axiommono} For all $f,g\in \mathcal{F}$, if, for all $s \in S$, $f(s) \succsim g(s)$, then $f \succsim g$.
\end{axiom}

In Axiom \ref{axiomtransiincomp_new}, we impose a relatively weak sufficient condition for comparability. 

\begin{axiom}\label{axiomtransiincomp_new}
For all $f, g \in \mathcal{F}$ and $x,y \in X$: if $f\Join x$, $x \succ g$, $g\Join y$, and $f\succ y$, then $f \succ g$. 
\end{axiom}

While the DM cannot compare $f$ to the constant $x$, she declares $x$ strictly more desirable than $g$. On the other hand, while she cannot compare $g$ to the constant act $y$, she declares $f$ strictly more desirable than $y$. Axiom \ref{axiomtransiincomp_new} implies that \textit{in the presence of such consonant conclusions as to the comparison of $f$ and $g$, the DM considers $f$, with sufficient conviction, strictly more desirable than $g$}.

As we already highlighted, given the complexity involved in the evaluation of uncertain acts, constant acts, which are the simplest acts, involving no ambiguity, are likely to be used as comparison devices. Then, a straightforward way to use them in comparing two acts, when preferences may be incomplete, consists in looking for a constant act that is incomparable to one of them and comparable to the other one. Each such constant act then provides \textit{a piece of evidence} as to the comparison between the two uncertain acts---the question is then to determine what are sufficient pieces of evidence.  
 
According to Axiom \ref{axiomtransiincomp_new}, two consonant pieces of evidence suffice: if there is a constant act $x$ that is incomparable to $f$ and more desirable than $g$, and a constant act $y$ that is incomparable to $g$ and less desirable than $f$, then both pieces of evidence favor $f$, and the DM declares $f$ more desirable than $g$.
 
There is a sense in which Axiom \ref{axiomtransiincomp_new} expresses, for an incomplete preference relation, a minimal departure from the completeness of weak orders for which all acts admit a \textit{certainty equivalent}.\footnote{That is, binary relations $\succsim$ which are reflexive, transitive and complete, such that, for all $f \in \mathcal{F}$, there is $x \in X$ such that $f \sim x$.} For these weak orders, one piece of evidence is sufficient: if $f \in \mathcal{F}$ has a certainty equivalent $x \in X$ and $x$ is strictly preferred to $g \in \mathcal{F}$, then $f$ is strictly preferred to $g$. For an incomplete preference relation, Axiom \ref{axiomtransiincomp_new} involves no more than one piece of evidence based on a constant act incomparable to $f$ \textit{and} one piece of evidence based on a constant act incomparable to $g$. 
\medskip

We sometimes refer to the classical Axioms \ref{axiomcont}, \ref{axiomcertain} and \ref{axiommono} as
continuity, certainty independence and monotonicity.

\begin{theorem}\label{theo:characterization}
A binary relation $\succsim$ satisfies Axioms \ref{axiomclassic}-\ref{axiomtransiincomp_new} if and only if there exist
\begin{itemize}
    \item [$\color{black} \bullet$] a non-constant affine function $u:X \to \mathbb{R}$, unique up to positive affine transformation,
    \item [$\color{black} \bullet$] a unique pair $(C,D)$ of non-disjoint convex compact subsets of $\Delta$,
\end{itemize} such that, for all $f,g \in \mathcal{F},$ 
\begin{align*}
f \succsim g \Leftrightarrow \begin{cases}
\min_{p\in C} \int u(f) dp \geq \min_{p\in C} \int u(g) dp \\[4pt]
\max_{p\in D} \int u(f) dp \geq \max_{p\in D} \int u(g) dp
\end{cases},
\end{align*}
that is, $\succsim$ admits the hope-and-prepare representation $(u,C,D)$, where $C$ and $D$ are unique, and $u$ is unique up to positive affine transformation.
\end{theorem}

In the proof of Theorem \ref{theo:characterization}, we define two relations on the bases of $\succsim$: for all $f, g \in \mathcal{F}$,
\begin{align*}
    f \succsim_p g &\iff \text{for all } x \in X,  g \succsim x \text{ implies } f \succsim x, \\
    f \succsim_o g &\iff \text{for all } x \in X,  x \succsim f \text{ implies } x \succsim g.
\end{align*}
These relations capture the \textit{pessimistic} and \textit{optimistic} assessments underlying $\succsim$: $f \succsim_p g$ whenever $f$ passes every threshold that $g$ passes, and $f \succsim_o g$ whenever every constant act that is at least as desirable as $f$ is also at least as desirable as $g$.

\color{black}

\subsection{Characterization of concordant hope-and-prepare preferences}

We now identify conditions under which $C = D$ in representation (\ref{introrep}). Under certainty independence (Axiom \ref{axiomcertain}), comparisons between acts are preserved under mixtures with constant acts. The following two axioms impose that comparisons between an act and \textit{a constant act} be preserved under mixtures with \textit{any} act.\bigskip

\begin{axiom}\label{axiompess}
For all $f\in \mathcal{F}$, $x \in X$, $h \in \mathcal{F}$, and $\alpha \in (0,1]$, $f \succsim x$ implies $\alpha f + (1-\alpha)h \succsim \alpha x + (1-\alpha) h.$
\end{axiom}
 
\begin{axiom}\label{axiomopt}
For all $f\in \mathcal{F}$, $x \in X$, $h \in \mathcal{F}$, and $\alpha \in (0,1]$, $x \succsim f$ implies $\alpha x + (1-\alpha)h \succsim \alpha f + (1-\alpha) h.$
\end{axiom}

Axiom \ref{axiompess} states that being at least as desirable as a constant act reflects an unambiguously favorable assessment: it cannot be overturned by mixing with any act. Axiom \ref{axiomopt} states symmetrically that being at most as desirable as a constant act reflects an unambiguously unfavorable assessment. 

For a hope-and-prepare preference relation, the two axioms ensure that the pessimistic and optimistic evaluations are computed on the same set of scenarios.\footnote{We thank an anonymous referee who suggested adopting these two axioms.}

\color{black}

\begin{theorem}\label{propconcor} The following statements hold:
\begin{enumerate}[(i)]
    \item A hope-and-prepare preference $\succsim$, with unique representation $(u,C,D)$, satisfies 
    Axiom \ref{axiompess} if and only if $D \subseteq C$.
    \item A hope-and-prepare preference $\succsim$, with unique representation $(u,C,D)$, satisfies  
    Axiom \ref{axiomopt} if and only if $C \subseteq D$.
\end{enumerate}
In particular, a binary relation $\succsim$ satisfies Axioms \ref{axiomclassic}-\ref{axiomopt} if and only if there exist

\begin{itemize}
    \item a non-constant affine function $u:X \to \mathbb{R}$, unique up to positive affine transformation,
    \item a unique convex compact subset of $\Delta$, denoted $C$,
such that, for all $f,g \in \mathcal{F}$,
\end{itemize}  
\begin{align*}
f \succsim g \iff \begin{cases}
\min_{p\in C} \int u(f) dp \geq \min_{p\in C} \int u(g) dp \\[4pt]
\max_{p\in C} \int u(f) dp \geq \max_{p\in C} \int u(g) dp
\end{cases}.
\end{align*}
\end{theorem}

When $\succsim$ admits a concordant representation, acts are evaluated on the basis of the interval of all expected utility levels that they induce across scenarios in a given set, according to the strong set order. An act $f$ is at least as desirable as $g$ if and only if there exists a level attainable from $f$ that is at least as high as every level attainable from $g$, and a level attainable from $g$ that is at most as high as every level attainable from $f$.

We note that the two additional axioms imposed in our characterization of concordant hope-and-prepare preferences are, arguably, simpler to interpret than the additional axioms imposed in \cite{E2022} in their characterization of concordant twofold preferences, which involve \textit{complementary acts} (\cite{S2009}).


\section{Hope-and-prepare preferences and the $\alpha$-MEU model}\label{sec:extension}

In this section, we explore the relationship between hope-and-prepare preferences and the (asymetric) $\alpha$-MEU model. 

\subsection{Invariant biseparable completion}

We focus on \textit{invariant biseparable} complete extensions of hope-and-prepare preferences . A complete preorder $\succsim^*$ on $\mathcal{F}$ is invariant biseparable in the sense of \cite{GM2004} if it satisfies Axioms \ref{axiomcont}, \ref{axiomcertain} and \ref{axiommono}. This class of preferences includes important models of complete preferences under uncertainty such as subjective expected utility, maxmin expected utility, and the $\alpha$-MEU model.

\begin{definition}\label{defalphamaxmin}
A preference relation $\succsim$ on $\mathcal{F}$ admits an \textit{asymmetric $\alpha$-MEU} representation if there exist $\alpha\in [0,1]$, two non-disjoint compact convex subsets $C$ and $D$ of $\Delta$, and a non-constant affine function $u:X \to \mathbb{R}$ such that for all $f,g \in \mathcal{F}$,
\begin{align*}
   f \succsim g \iff \ &\alpha \min_{p\in C}\int u(f)dp + (1-\alpha) \max_{p\in D} \int u(f) dp \\
   &\geq \alpha \min_{p\in C}\int u(g)dp + (1-\alpha) \max_{p\in D} \int u(g) dp.
\end{align*}
We refer to such representation as a $(u,C,D,\alpha)$ representation.
\end{definition}

Remarkably, \cite{CE2022} show that the asymmetric $\alpha$-MEU, while retaining the tractability property of the standard $\alpha$-MEU, is flexible enough to accommodate ambiguity aversion for \textit{large/moderate-likelihood} events but ambiguity seeking for \textit{small-likelihood events} and \textit{source-dependent} ambiguity attitudes.

Standard $\alpha$-MEU criteria are obtained if $C=D$ in Definition \ref{defalphamaxmin}, and the following result, as a particular case, characterizes them as invariant biseparable extensions of concordant hope-and-prepare preferences.

\begin{theorem}\label{theorem-extention}
The following conditions are equivalent when $\succsim$ is a hope-and-prepare preference with unique representation $(u,C,D)$:
\begin{enumerate}[(i)]
    \item $\succsim^*$ is an invariant biseparable preference and an extension of $\succsim$.
    \item $\succsim^*$ admits an asymmetric $\alpha$-maxmin expected utility representation $(u,C,D,\alpha)$ in which $\alpha$ is unique whenever $\succsim$ is not complete. 
\end{enumerate}
\end{theorem}

Beyond the fact that we consider the completion of a new type of preferences, our characterization has three salient features compared to other inquiries about the $\alpha$-MEU criterion: the representation is asymmetric in general, $\alpha$ does not depend on the considered acts, and is essentially unique.

\subsection{Recovering the hope-and-prepare preference from its completion}
 
We focus on concordant hope-and-prepare preferences, whose invariant biseparable completions are standard $\alpha$-MEU preferences (Theorem \ref{theorem-extention}).\footnote{For a non-concordant hope-and-prepare preference the completion, is an asymmetric $\alpha$-MEU preference (Definition \ref{defalphamaxmin} with $C \neq D$). The construction below relies on a common set of priors and does not extend to this case, which we leave open.} Throughout this subsection, $\unrhd$ denotes a standard $\alpha$-MEU preference with representation $(u, C, \alpha)$ for some $\alpha \in [0,1]$. We ask the converse of the extension theorem: given $\unrhd$, can the hope-and-prepare preference it completes be expressed in terms of $\unrhd$ itself? 

The standard $\alpha$-MEU representation is not unique: distinct pairs $(C, \alpha)$ can induce the same ranking \citep{FRY2022, HARTMANN2023105719}. However, for a given $\alpha \neq 1/2$, the set of priors in the $\alpha$-MEU representation is unique \citep[Theorem 1]{HARTMANN2023105719}. Therefore, as in \citet{HARTMANN2023105719}, we fix the value of $\alpha$, assumed to belong to $(0,1) \setminus \{1/2\}$.\footnote{When $\alpha = 1/2$, there might exist a set of priors $C' \neq C$ such that the $\alpha$-MEU representation $(u,C',\alpha)$ represents $\unrhd$.} 
 
Recall that two acts $f$ and $\bar{f}$ are \textit{complementary} if their equal-weight mixture is equivalent to a constant act: $\frac{1}{2}f(s) + \frac{1}{2}\bar{f}(s) \sim \frac{1}{2}f(s') + \frac{1}{2}\bar{f}(s')$ for all $s, s' \in S$ \citep{S2009}. Given $\unrhd$, define the \textit{dual preference} $\unrhd_d$ in the following way: for all $f, g \in \mathcal{F}$,
$$f \unrhd_d g \iff \bar{g} \unrhd \bar{f},$$
where $\bar{f}$ and $\bar{g}$ are complementary to $f$ and $g$ respectively, with $\frac{1}{2}f + \frac{1}{2}\bar{f} \sim x \sim \frac{1}{2}g + \frac{1}{2}\bar{g}$ for some constant act $x$.\footnote{It is easy to prove that such a constant act always exists. 
} 

The next result shows that $\unrhd_d$ is well-defined.
 
\begin{lemma}\label{lemma:dual}
Let $\unrhd$ admit a standard $\alpha$-MEU representation $(u, C, \alpha)$. Then $\unrhd_d$ is represented by the mapping $V_d: \mathcal{F} \to \mathbb{R}$ defined by
$$V_d(h) = (1-\alpha) \min_{p \in C} \int u(h)dp + \alpha \max_{p \in C} \int u(h)dp.$$
In particular, for all $f,g \in \mathcal{F}$, $f \unrhd_d g$ is independent of the choice of complementary acts $\bar{f}, \bar{g}$ and of the constant act $x$.
\end{lemma}
 
By Lemma \ref{lemma:dual}, $\unrhd_d$ is a $(1-\alpha)$-MEU preference. For each act $f$, let $c_f$ denote the certainty equivalent of $f$ under $\unrhd$ and denote $z_f$ the certainty equivalent of $f$ under $\unrhd_d$. 
 
\begin{proposition}\label{prop:recovery}
Let $\unrhd$ admit a standard $\alpha$-MEU representation $(u, C, \alpha)$ with $\alpha \in \left(\frac{1}{2}, 1\right)$, and let $\unrhd_d$ be its dual preference. Let $\succsim^\circ$ be a concordant hope-and-prepare preference whose invariant biseparable completion is $\unrhd$. Then, $\succsim^\circ$ is unique and, for all $f, g \in \mathcal{F}$, 
$$f \succsim^\circ g \iff \begin{cases} \alpha \cdot c_f + (1-\alpha) \cdot z_g \unrhd \alpha \cdot c_g + (1-\alpha) \cdot z_f \\[4pt] \alpha \cdot z_f + (1-\alpha) \cdot c_g \unrhd \alpha \cdot z_g + (1-\alpha) \cdot c_f \end{cases}.$$
\end{proposition}
 
The case $\alpha \in (0, 1/2)$ is symmetric, with the roles of $\unrhd$ and $\unrhd_d$ exchanged. 
 
The act $f$ is summarized by two certainty equivalents, $c_f$ under $\unrhd$ and $z_f$ under $\unrhd_d$. The $\alpha$-MEU evaluation underlying $c_f$ weights the worst-case value by $\alpha$ and the best-case value by $1-\alpha$, and the dual evaluation underlying $z_f$ reverses these weights. Comparing $c_f$ with $c_g$ reveals only the $\alpha$-MEU ranking, not the separate worst-case and best-case rankings. These are recovered by mixing the certainty equivalents \textit{across} acts. That is, comparing $\alpha \cdot c_f + (1-\alpha) \cdot z_g$ with $\alpha \cdot c_g + (1-\alpha) \cdot z_f$ cancels the best-case terms and isolates the worst-case ranking. Symmetrically, comparing $\alpha \cdot z_f + (1-\alpha) \cdot c_g$ with $\alpha \cdot z_g + (1-\alpha) \cdot c_f$ isolates the best-case ranking.
 
Proposition \ref{prop:recovery} takes the value of $\alpha$ as given, and each possible value of $\alpha$ is associated with a specific set of priors in the $\alpha$-MEU representation. Thus, the construction in Proposition \ref{prop:recovery} expresses, for each $\alpha$, the hope-and-prepare preference with the corresponding set of priors, rather than deriving a single hope-and-prepare preference simply from $\unrhd$. It would not be possible to do the latter because $\unrhd$ may be consistent with different pairs $(\alpha, C)$ and $(\alpha', C')$.\footnote{More precisely, as shown in \cite{HARTMANN2023105719} (Theorem 2), the values of $\alpha$ that are compatible with $\unrhd$ form an interval, and the associated sets of priors are nested.}

\color{black}

\section{Comparison of incomplete criteria}\label{sec:comp}

\subsection{Degree of incompleteness}\label{subsec:comp}

We have stated that with hope-and-prepare preferences, in comparison to Bewley preferences and twofold preferences, the
trade-off between decisiveness and conviction is addressed in a way that is more favorable to decisiveness. The criterion we use to determine whether a binary relation is more conservative than an other one pertains to their respective \textit{degree of incompleteness}. 

\begin{definition}\label{moreconserv}
Given two preference relations $\succsim_1$ and $\succsim_2$ on $\mathcal{F}$, we say that $\succsim_1$ is more conservative than $\succsim_2$ if $\succsim_2$ is an extension of $\succsim_1$, that is, for all $f,g \in \mathcal{F}$,
$$f \succsim_1 g \text{ implies } f\succsim_2 g.$$
\end{definition}

The next proposition identifies necessary and sufficient conditions under which a hope-and-prepare preference relation is an extension of a Bewley or of a twofold preference relation.

\begin{proposition}\label{prop_comp_Bew_Two} Let $\succsim_{HP}$, $\succsim_T$, and $\succsim_B$ be, respectively, a hope-and-prepare preference with unique representation $(u,C_{HP},D_{HP})$, a twofold preference with unique representation $(u,C_T,D_T)$, and a Bewley preference with unique representation $(u,C_B)$, sharing the same utility function $u$. Then,
\begin{enumerate}
\item[(i)] $\succsim_B$ is more conservative than $\succsim_{HP}$ if and only if $C_{HP} \cup D_{HP} \subseteq C_B$;
\item[(ii)] $\succsim_T$ is more conservative than $\succsim_{HP}$ if and only if $C_{HP} \subseteq C_T$ and $D_{HP} \subseteq D_T$.
\end{enumerate} 
\end{proposition}

\begin{remark}
A direct consequence of this proposition and Proposition 4 in \cite{E2022}  is that if $C_{HP} \cup D_{HP} \subseteq C_B \subseteq C_T \cap D_T$, in particular if $C_{HP} = D_{HP} = C_B = C_T = D_T$, then $\succsim_T$ is more conservative than $\succsim_B$, which is more conservative than $\succsim_{HP}$.
\end{remark}

When comparing two options under ambiguity, a Bewley DM requires \textit{every} scenario to agree on the ranking---a demanding standard that leaves many pairs unranked. A twofold DM goes further: she requires the worst case of one option to beat the best case of the other. A hope-and-prepare DM asks less: she only needs the worst cases to agree among themselves, and the best cases to agree among themselves. By separating the two assessments rather than requiring one to dominate the other, she can resolve more comparisons while still demanding that both perspectives point in the same direction.

Proposition \ref{prop_comp_Bew_Two} identifies the conditions on the sets of priors under which these comparisons hold. In particular, when all three DMs share the same utility and draw on the same set of scenarios, the twofold preference is more conservative than the Bewley preference, which is in turn more conservative than the hope-and-prepare preference: each step gains decisiveness by relaxing what ``agreement across scenarios'' means.

\subsection{Ambiguity attitudes}\label{sectionambatt}

We are able to compare ambiguity attitudes displayed by different hope-and-prepare preferences using the classical comparative statics notions of \cite{GM2002}.

\begin{definition}
Given two preference relations $\succsim_1$ and $\succsim_2$ on $\mathcal{F}$,
\begin{enumerate}[(i)]
    \item $\succsim_1$ is more ambiguity averse than $\succsim_2$ if, for all $f \in \mathcal{F}$ and $x \in X$, $f \succsim_1 x$ implies $f \succsim_2 x$.
    \item $\succsim_1$ is more ambiguity loving than $\succsim_2$ if, for all $f \in \mathcal{F}$ and $x \in X$, $x \succsim_1 f$ implies $x \succsim_2 f$.
\end{enumerate}
\end{definition}

An agent is more ambiguity averse than another if she is less inclined to consider an uncertain act $f$ at least as desirable as a constant act $x$. On the other hand, an agent is more ambiguity loving than another if she is more inclined to consider an uncertain act $f$ at least as desirable as a constant act $x$. The next result characterizes ambiguity attitudes for hope-and-prepare preferences.

\begin{proposition}\label{Prop_comparative_attitude}
Let $\succsim_1$ and $\succsim_2$ be two hope-and-prepare preference relations with unique representations $(u,C_1,D_1)$ and $(u,C_2, D_2)$, respectively. Then,
\begin{enumerate}[(i)]
\item $\succsim_1$ is more ambiguity averse than $\succsim_2$ if and only if $C_2 \subseteq C_1$.
\item $\succsim_1$ is more ambiguity loving than $\succsim_2$ if and only if $D_2 \subseteq D_1$.
\end{enumerate}
\end{proposition}

For a hope-and-prepare representation $(u,C,D)$, the two sets of priors $C$ and $D$ represent the level of pessimism and optimism related to the DM's ambiguity attitudes. More precisely, the relationship $C_2 \subseteq C_1$ means that, in the worst scenario, the level of welfare attained by the agent if she has preference relation $\succsim_1$ is lower than the one attained if she has preference relation $\succsim_2$. Similarly, $D_2 \subseteq D_1$ means that, in the best scenario, the level of welfare attained by the agent if she has preference relation $\succsim_1$ is higher than the one attained if she has preference relation $\succsim_2$.

Based on Proposition \ref{Prop_comparative_attitude} \textit{(i)}, by comparing the concordant preference $\succsim$ with representation $(u,C,C)$ to the non-concordant preference $\succsim_1$ with representation $(u,C_1,C)$, with $C_1 \subset C$, we can say that $\succsim_1$ \textit{is more ambiguity averse than it is ambiguity loving}. Similarly, the non-concordant preference $\succsim_2$ with representation $(u,C, D_2)$, with $D_2 \subset C$, \textit{can be said to be more ambiguity loving than it is ambiguity averse}. A DM with concordant preferences is as ambiguity loving as she is ambiguity averse, or, in other words, \textit{her pessimistic evaluation is as pessimistic as her optimistic evaluation is optimistic}.

\section{Choquet evaluations}\label{sec:choquet}

In this section, we extend the idea of combining a pessimistic and an optimistic evaluation to preference relations that may violate certainty independence (Axiom \ref{axiomcertain}). Essentially, Axiom \ref{axiomcertain} is replaced by the comonotonic independence condition of \cite{S1989}, and convexity (Axiom \ref{axiomconv}) is dropped. Then, the characterized representation involves two evaluations, each performed according to a Choquet integral, one of them ---the pessimistic evaluation--- being always below the other. Finally, if the preference relation is required to satisfy Axiom \ref{axiomconv}, then it must be a hope-and-prepare preference relation.

We impose a continuity axiom which is stronger than Axiom \ref{axiomcont}.

\begingroup
\renewcommand{\theaxiom}{2$'$}
\begin{axiom}[Continuity]\label{axiomcontprime}
For all $f,g,f',g'\in\mathcal F$, the set $\{\lambda\in[0,1]:\lambda f+(1-\lambda)g \succsim \lambda f'+(1-\lambda)g'\}$ is closed.
\end{axiom}
\endgroup
\addtocounter{axiom}{-1}

Axiom \ref{axiomcontprime} refers only to the mixture structure of $\mathcal F$; no topology on $\mathcal F$ is presumed.\footnote{If one endows $X$ with a topology for which mixtures and $u$ are continuous, Axiom \ref{axiomcontprime} is equivalent to closedness of $\{(f,g):f\succsim g\}$ in the product topology on $\mathcal F\times\mathcal F$.} 

Recall that $f,g\in\mathcal F$ are \textit{comonotonic} if there are no $s,t\in S$ with $f(s)\succ f(t)$ and $g(t)\succ g(s)$. We impose
the following condition in place of Axiom \ref{axiomcertain}, retaining
Axioms \ref{axiomclassic}, \ref{axiommono}, and \ref{axiomtransiincomp_new}.

\begingroup
\renewcommand{\theaxiom}{3$'$}
\begin{axiom}[Comonotonic independence]\label{axiomcomono}
For all pairwise comonotonic $f,g,h\in\mathcal F$ and all $\alpha\in(0,1)$,
$f\succsim g$ if and only if $\alpha f+(1-\alpha)h\succsim\alpha g+(1-\alpha)h$.
\end{axiom}
\endgroup
\addtocounter{axiom}{-1}

The following result shows that replacing certainty independence with comonotonic independence implies that the DM uses two evaluations based on Choquet integrals, one of the evaluations always lying below the other.

For $\nu$ a capacity on $(S,\Sigma)$, we let $\varphi \in B_0(\Sigma) \mapsto \int\varphi  d\nu \in \mathbb{R}$ denote the associated Choquet integral.\footnote{A (normalized) capacity on $(S,\Sigma)$ is a set function $\nu$ with $\nu(\emptyset)=0$, $\nu(S)=1$,
and $\nu(A)\le\nu(B)$ whenever $A\subseteq B$. Then, 
\[
\int \varphi \, d\nu \;=\; \int_0^{\infty} \nu\bigl(\{\varphi \ge t\}\bigr)\, dt \;+\; \int_{-\infty}^{0} \Bigl[\nu\bigl(\{\varphi \ge t\}\bigr) - 1\Bigr]\, dt,
\]
where both integrals on the right-hand side are Riemann integrals and $\{\varphi \ge t\} := \{s \in S : \varphi(s) \ge t\}$.}

\begin{theorem}\label{thm:choquet}
A binary relation $\succsim$ satisfies Axioms \ref{axiomclassic}, \ref{axiomcontprime}, \ref{axiomcomono}, \ref{axiommono}, and \ref{axiomtransiincomp_new} if and only if there exist a non-constant affine function $u:X\to\mathbb R$ and capacities $\nu_p,\nu_o$ on $(S,\Sigma)$ such that $\int u\circ fd\nu_p\le\int u\circ fd\nu_o$ for all $f\in\mathcal F$, and for all $f,g\in\mathcal F$,
\begin{equation*}
f\succsim g \quad \Longleftrightarrow\quad \int u\circ fd\nu_p\ge\int u\circ gd\nu_p \ \text{ and }\ \int u\circ fd\nu_o\ge\int u\circ gd\nu_o .
\end{equation*}
The function $u$ is unique up to a positive affine transformation, and, given $u$, the capacities $\nu_p$ and $\nu_o$ are unique.
\end{theorem}

The inequality $\int u\circ fd\nu_p\le\int u\circ fd\nu_o$ is the analogue, for capacities, of the non-disjointness of the prior sets in Theorem \ref{theo:characterization}: the Choquet integral with respect to $\nu_p$ can be interpreted as a pessimistic evaluation and that with respect to $\nu_o$ as an optimistic one. Because a capacity need not be supermodular or submodular, \textit{each} evaluation can accommodate ambiguity attitudes that vary across events---for instance the overweighting of unlikely extreme events---rather than a single attitude. 

The convexity property imposed in Axiom \ref{axiomconv} prevents this feature and yields a hope-and-prepare preference representation.

\begin{corollary}\label{cor:choquetconvex}
Suppose $\succsim$ satisfies Axioms \ref{axiomclassic}, \ref{axiomcontprime}, \ref{axiomcomono} \ref{axiommono}, and
\ref{axiomtransiincomp_new}, and let $u,\nu_p,\nu_o$ be as in Theorem \ref{thm:choquet}. Then $\succsim$ satisfies Axiom \ref{axiomconv} if and only if $\nu_p$ is supermodular and $\nu_o$ is submodular. In this case,
\begin{equation*}
\int u\circ fd\nu_p=\min_{p\in\mathrm{core}(\nu_p)}\int u\circ fdp, \qquad \int u\circ fd\nu_o=\max_{p\in\mathrm{anti\text{-}core}(\nu_o)}\int u\circ fdp,
\end{equation*}
that is, $\succsim$ is a hope-and-prepare preference with representation $\big(u,\mathrm{core}(\nu_p),\mathrm{anti\text{-}core}(\nu_o)\big)$.\footnote{A capacity $\nu$ is \textit{supermodular} if, for all $A,B \in \Sigma$,
$\nu(A\cup B)+\nu(A\cap B)\ge \nu(A)+\nu(B)$. It is \textit{submodular} if for all $A,B\in\Sigma$, $
\nu(A\cup B)+\nu(A\cap B)\le \nu(A)+\nu(B)$. The core and anti-core of a capacity $\nu$ on $(S,\Sigma)$ are $\mathrm{core}(\nu) = \{p \in \Delta : p(A) \geq \nu(A) \text{ for all } A \in \Sigma\}$ and $\mathrm{anti\text{-}core}(\nu) = \{p \in \Delta : p(A) \leq \nu(A) \text{ for all } A \in \Sigma\}$.}
\end{corollary}

The proof of this basic result is omitted.

\section{Conclusion}\label{sec:conclusion}

We provided a new perspective on the analysis of incomplete preferences under uncertainty by introducing and characterizing a new decision criterion involving multiple priors. It is based on a requirement of unanimity between an optimistic and a pessimistic evaluation reflecting the behavior of a DM who \textit{hopes for the best while she also prepares for the worst.} When both of these evaluations are computed according to the same set of scenarios, hope-and-prepare preferences compare ranges of expected utility according to the strong set order. 

Comparing hope-and-prepare preferences to the two closest incomplete criteria proposed in this framework---Bewley and twofold preferences---we argued, and made visible in our axiomatization, that the trade-off between decisiveness and conviction is addressed in a way that is more favorable to decisiveness.

We showed that an \textit{invariant biseparable completion} of a hope-and-prepare preference relation necessarily admits a \textit{unique} $\alpha$-MEU representation. We also established a converse: when the $\alpha$-MEU representation is unique,  the underlying hope-and-prepare preference can be recovered from the complete preference and its dual preference relation, whose construction involves complementary acts.

\section*{Appendix: Proofs}
\appendix

We let $B_{0}(\Sigma)$ denote the set of all real-valued, $\Sigma$-measurable simple functions on $S$. Because every act $f$ is simple, $u \circ f \in B_{0}(\Sigma)$ for every utility index $u: X \to \mathbb{R}$. A functional $I: B_{0}(\Sigma) \to \mathbb{R}$ is said to be \textit{constant-linear} if
\[
  I(a\varphi + b) = a\,I(\varphi) + b
  \quad \text{for all } \varphi \in B_{0}(\Sigma),\ a \in \mathbb{R}_+,\ b \in \mathbb{R},
\]
where, abusing notation, $b$ also denotes the constant function $s \in S \mapsto b$. It is said to be \textit{monotone} if $I(\varphi) \ge I(\phi)$ for all $\varphi, \phi \in B_{0}(\Sigma)$ with $\varphi \ge \phi$.

\section{Proof of Theorem \ref{theo:characterization}}

\noindent\textbf{Sufficiency.} Assume $\succsim$ admits the representation $(u,C,D)$. Axioms \ref{axiomclassic}--\ref{axiommono} are readily verified. We show that Axiom \ref{axiomtransiincomp_new} holds. 

Under the representation, for any $f \in \mathcal{F}$ and $x \in X$: $f \succsim x$ if and only if $\min_{p \in C}\int u(f)dp \geq u(x)$, and $x \succsim f$ if and only if $u(x) \geq \max_{p \in D}\int u(f)dp$. Thus, $f \Join x$ if and only if
$$\min_{p \in C}\int u(f)dp < u(x) < \max_{p \in D}\int u(f)dp.$$

Let $f,g \in \mathcal{F}$ and $x,y \in X$ be such that $f \Join x$, $x \succ g$, $g \Join y$, and $f \succ y$, as assumed in Axiom \ref{axiomtransiincomp_new}. One has
\begin{alignat*}{2}
    &\min_C \int u(f)dp < u(x) < \max_D \int u(f)dp, &\qquad &(\text{from } f \Join x) \\
    &u(x) \geq \max_D \int u(g)dp, &\qquad &(\text{from } x \succ g) \\
    &\min_C \int u(g)dp < u(y) < \max_D \int u(g)dp, &\qquad &(\text{from } g \Join y) \\
    &\min_C \int u(f)dp \geq u(y). &\qquad &(\text{from } f \succ y)
\end{alignat*}

Combining the third and the fourth lines yields $\min_C \int u(f)dp \geq u(y) > \min_C \int u(g)dp$. Combining the first and second yields $\max_D \int u(f)dp > u(x) \geq \max_D \int u(g)dp$. We have proved that $f \succ g$.

\medskip
\noindent\textbf{Necessity.} Assume that $\succsim$ satisfies Axioms \ref{axiomclassic}--\ref{axiomtransiincomp_new}. Since the restriction of $\succsim$ to $X$ is a non-trivial complete preorder satisfying certainty independence and continuity, there exists a non-constant affine function $u: X \to \mathbb{R}$, unique up to positive affine transformation, such that $x \succsim y$ if and only if $u(x) \geq u(y)$ for all $x, y \in X$.

From $\succsim$, define two binary relations on $\mathcal{F}$:
\begin{align*}
    f \succsim_p g &\iff \text{for all } x \in X,  g \succsim x \text{ implies } f \succsim x, \\
    f \succsim_o g &\iff \text{for all } x \in X,  x \succsim f \text{ implies } x \succsim g.
\end{align*}

It is clear that these relations are preorders.

\begin{step}\label{step:wo}
$\succsim_p$ and $\succsim_o$ are complete.
\end{step}

We prove the claims for $\succsim_p$; the argument for $\succsim_o$ is symmetric. By contradiction, assume that there are $f$ and $g$ in $\mathcal{F}$ such that $f \varnotsuccsim_p g$ and $g \varnotsuccsim_p f$. Then there exist $x, y \in X$ such that $g \succsim x$ but $f \varnotsuccsim x$, and $f \succsim y$ but $g \varnotsuccsim y$. Since $\succsim$ is complete on $X$, either $x \succsim y$ or $y \succ x$. If $x \succsim y$, then $g \succsim x \succsim y$; contradicting $g \varnotsuccsim y$. If $y \succ x$, then $f \succsim y \succsim x$; contradicting $f \varnotsuccsim x$.

\begin{step}\label{step:connection} Let $f, g \in \mathcal{F}$ be such that $f \succsim g$. Then, $f \succsim_p g$ and $f \succsim_o g$.
\end{step}

Let $x \in X$ be such that $g \succsim x$. By transitivity, $f \succsim x$. Thus, $f \succsim_p g$. Let $y \in X$ be such that $y \succsim f$. Then $y \succsim g$. Thus, $f \succsim_o g$.

\begin{step}\label{step:properties}
$\succsim_p$ and $\succsim_o$ satisfy continuity, certainty independence, and monotonicity.
\end{step}

We prove the claims for $\succsim_p$; the proofs for $\succsim_o$ are symmetric.

\smallskip\noindent\textit{Continuity.} Let $f, g, h \in \mathcal{F}$. By definition of $\succsim_p$, for all $\alpha \in [0,1]$, $\alpha f + (1-\alpha)g \succsim_p h$ if and only if, for all $x \in X$ such that $h \succsim x$, one has $\alpha f + (1-\alpha)g \succsim x$. Therefore,
$$\{\alpha \in [0,1] : \alpha f + (1-\alpha)g \succsim_p h\} = \bigcap_{x \in X: h \succsim x} \{\alpha \in [0,1] : \alpha f + (1-\alpha)g \succsim x\}.$$
Each set in the intersection is closed by Axiom \ref{axiomcont}, their intersection is thus closed. That $\{\alpha \in [0,1] : h \succsim_p \alpha f + (1-\alpha)g\}$ is closed is proved analogously. We conclude that $\succsim_p$ is continuous.

\smallskip\noindent\textit{Certainty independence.} As a recall, we must prove that for all $f,g \in \mathcal{F}$, $f \succsim_p g$ if and only if [$\alpha f + (1-\alpha)z \succsim_p \alpha g + (1-\alpha)z$, for all $\alpha \in (0,1)$ and $z \in X$]. The sufficiency direction immediately follows from the continuity of $\succsim_p$.

We now prove the necessity direction. Let $f,g \in \mathcal{F}$ be such that $f \succsim_p g$. Let $z \in X$ and let $x \in X$ be such that $\alpha g + (1-\alpha)z \succsim x$, for some $\alpha \in (0,1)$. 

First, assume that there exists $x_* \in X$ such that $x \succsim \alpha x_*+(1-\alpha)z$. By definition of a simple act, there is $x^*$ such that, for all $s \in S$, $x^* \succsim g(s)$. Consider the sets 
$$\{\beta \in [0,1]: \alpha [\beta x^*+(1-\beta)x_*]+(1-\alpha)z \succsim x\}$$ and $$\{\beta \in [0,1]: x \succsim \alpha [\beta x^*+(1-\beta)x_*]+(1-\alpha)z\}.$$ By Axiom \ref{axiommono}, they are non-empty. They are closed relative to $[0,1]$ by the continuity of $\succsim$; and they cover $[0,1]$ since $\succsim$ is complete on $X$. The connectedness of $[0,1]$ in turn implies that their intersection is non-empty: there is $w \in X$ such that $\alpha w+(1-\alpha)z\sim x$. Since, $\alpha g + (1-\alpha)z \succsim \alpha w + (1-\alpha)z$, Axiom \ref{axiomcertain} implies $g \succsim w$. Since $f \succsim_p g$, one has $f \succsim w$. By Axiom \ref{axiomcertain} again, $\alpha f + (1-\alpha)z \succsim \alpha w + (1-\alpha)z$, which yields $\alpha f + (1-\alpha)z \succsim x$ by transitivity.

Now, assume that for all $y \in X$, $\alpha y+(1-\alpha)z \succ x$. This holds in particular for $y_* \in X$ such that, for all $s \in S$, $f(s) \succsim y_*$. Then, by Axiom \ref{axiommono}, $\alpha f+(1-\alpha)z \succsim x$.

We have proved that $\alpha f + (1-\alpha)z \succsim_p \alpha g + (1-\alpha)z$.

\smallskip\noindent\textit{Monotonicity.} This property immediately follows from Step \ref{step:connection}.

\begin{step}\label{step:aversion}
$\succsim_p$ displays ambiguity aversion, i.e., for all $f, g\in \mathcal{F}$, $f\sim_p g$ implies $\alpha f +(1-\alpha)g \succsim_p f$; $\succsim_o$ displays ambiguity loving, i.e., for all $f, g\in \mathcal{F}$, $f\sim_p g$ implies $f \succsim_o \alpha f +(1-\alpha)g$.
\end{step}

We prove that $\succsim_p$ displays ambiguity aversion. Let $f, g\in \mathcal{F}$ be such that $f\sim_p g$. Let $x \in X$ be such that $f \succsim x$. Then, $g \sim_p f$ implies $g \succsim x$. Both $f$ and $g$ belong to the set $\{h \in \mathcal{F} : h \succsim x\}$, which is convex by Axiom \ref{axiomconv}, so $\alpha f + (1-\alpha)g \succsim x$. We have proved that $\alpha f + (1-\alpha)g \succsim_p f$. 

The argument for $\succsim_o$ is symmetric, using the convexity of $\{h \in \mathcal{F} : x \succsim h\}$.

\bigskip
By \cite{GS1989}, $\succsim_p$ admits a representation $f \mapsto \min_{p \in C} \int u_p (f) dp$ and $\succsim_o$ admits a representation $f \mapsto \max_{p \in D} \int u_o (f) dp$, where $C$ and $D$ are unique non-empty convex compact subsets of $\Delta$, and $u_p$ and $u_o$ are numerical representations of $\succsim_p$ and $\succsim_o$ on $X$, respectively. Since both $\succsim_p$ and $\succsim_o$ are consistent with $\succsim$ on $X$, the representing representations coincide up to positive affine transformation, and we may take $u_p = u_o = u$.  

\begin{step}
$C$ and $D$ are non-disjoint.   
\end{step}

\textit{Claim: $C \cap D \neq \emptyset$. if, and only if, for all $f \in \mathcal{F}$, $\min_{p \in C} \int u(f)dp \leq \max_{p \in D} \int u(f)dp$.}

We only prove the \textit{if part}, the other direction being trivial. We proceed by contraposition. Suppose that $C \cap D = \emptyset$. By the separating hyperplane theorem, there exists a bounded measurable function $\varphi : S\to \mathbb{R}$ such that $\min_{p\in C} \int \varphi dp > \max_{p\in D}\int \varphi dp$. Yet, there exists a sequence of simple functions $\{\varphi_n\}$ that converges (in supnorm topology) to $\varphi$. Since both $\Tilde{\varphi} \mapsto \min_{p\in C} \int \Tilde{\varphi} dp$ and $\Tilde{\varphi} \mapsto \max_{p\in D}\int \Tilde{\varphi} dp$ are continuous, there is $n\in \mathbb{N}$ such that $\min_{p\in C} \int \varphi_n dp > \max_{p\in D}\int \varphi_n dp$. As $a\varphi_n +b$ also satisfies this last inequality for all $a>0$ and $b\in \mathbb{R}$, one can choose $a>0$ and $b \in \mathbb{R}$ such that $a\varphi_n(s) +b \in u(X)$ for all $s\in S$, which implies $\varphi_n = u(f)$ for some $f\in \mathcal{F}$: $$\min_{p\in C} \int u(f) dp > \max_{p\in D}\int u(f) dp.$$

Now, we conclude this step by showing that for all $f$ and $g$ in $\mathcal{F}$, $\min_{p \in C} \int u(f) dp \leq \max_{p \in D} \int u(f) dp$. Assume by contradiction that for some $f \in \mathcal{F}$, $\min_{p \in C} \int u(f) dp > \max_{p \in D} \int u(f) dp$. Since $u(X)$ is an interval, there exists $x \in X$ such that
$$\min_{p \in C} \int u(f) dp > u(x) > \max_{p \in D} \int u(f) dp.$$
The left inequality yields $f \succ_p x$. By definition of $\succsim_p$, this implies $f \succsim x$, which, by Step \ref{step:connection}, yields $f \succsim_o x$. However, the right inequality yields $x \succ_o f$; a contradiction.

\paragraph{Conclusion.}\label{step:closure}
Let $f,g \in \mathcal{F}$. Then, $f \succsim g$ if and only if $\min_{p\in C} \int u(f) dp \geq \min_{p\in C} \int u(g) dp$ and $\max_{p\in D} \int u(f) dp \geq \max_{p\in D} \int u(g) dp$.\bigskip

\noindent\textit{Necessity direction.} It immediately follows from Step \ref{step:connection}.

\smallskip\noindent\textit{Sufficiency direction.} Assume that $\min_C \int u(f) dp \geq \min_C \int u(g) dp$ and $\max_D \int u(f) dp \geq \max_D \int u(g) dp$. Since $u$ is non-constant, pick $x_0 \in X$ with $u(x_0) < \sup u(X)$. By Axiom \ref{axiomcertain}, $f \succsim g$ if and only if $\tfrac{1}{2}f + \tfrac{1}{2}x_0 \succsim \tfrac{1}{2}g + \tfrac{1}{2}x_0$. Let $\hat{f} = \tfrac{1}{2}f + \tfrac{1}{2}x_0$ and $\hat{g} = \tfrac{1}{2}g + \tfrac{1}{2}x_0$. Then,
$$\max\left(\max_D \int u(\hat{f}) dp, \max_D \int u(\hat{g}) dp\right) \leq \tfrac{1}{2}\sup u(X) + \tfrac{1}{2}u(x_0) < \sup u(X).$$
Since $u(X)$ is an interval, there exists $x^* \in X$ such that $$u(x^*) > \max\left\{\max_D \int u(\hat{f}) dp, \max_D \int u(\hat{g}) dp\right\}.$$ For all $n \geq 1$, define $\hat{f}_n = \left(1 - \tfrac{1}{n}\right)\hat{f} + \tfrac{1}{n}x^*$. Then,
\begin{align*}
    \min_C \int u(\hat{f}_n) dp &= \left(1-\tfrac{1}{n}\right)\min_C \int u(\hat{f}) dp + \tfrac{1}{n} u(x^*) > \min_C \int u(\hat{g}) dp, \text{and } \\
    \max_D \int u(\hat{f}_n) dp &= \left(1-\tfrac{1}{n}\right)\max_D \int u(\hat{f}) dp + \tfrac{1}{n} u(x^*) > \max_D \int u(\hat{g}) dp, 
\end{align*}
that is, $\hat{f}_n \succ_p \hat{g}$ and $\hat{f}_n \succ_o \hat{g}$. 

Since $\hat{f}_n \succ_p \hat{g}$, there exists $y \in X$ such that $\hat{f}_n \succsim y$ and $\hat{g} \varnotsuccsim y$. Since $\hat{f}_n \succ_o \hat{g}$, there exists $x \in X$ with $x \succsim \hat{g}$ and $x \varnotsuccsim \hat{f}_n$. 

Several cases must be distinguished. First, if $\hat{f}_n \succsim x$, then $\hat{f}_n \succsim \hat{g}$, by transitivity. Since $\hat{g} \succsim \hat{f}_n$ would yield $x \succsim \hat{f}_n$, which is a contradiction, one has $\hat{f}_n \succ \hat{g}$. Second, by a similar argument, if $y \succsim \hat{g}$, then $\hat{f}_n \succ \hat{g}$. The third case is $\hat{f}_n \varnotsuccsim x$ and $y \varnotsuccsim \hat{g}$, \textit{i.e.} $\hat{f}_n \Join x$ and $\hat{g} \Join y$. The comparison $\hat{g} \succsim x$ would imply $\hat{f}_n \succsim x$, as $\hat{f}_n \succsim_p g$; we thus conclude that $x \succ \hat{g}$. In addition, $\hat{f}_n \succ y$. Indeed,  by the completeness of $\succsim$ on $X$, either $x \succsim y$ or $y \succ x$, and both cases are incompatible with $y \succsim \hat{f}_n$.\footnote{The first case implies $x \succsim \hat{f}_n$ and the second case implies $\hat{f}_n \succsim x$, both contradictions.} Therefore, $\hat{f}_n \Join x$, $x \succ \hat{g}$, $\hat{g} \Join y$, and $\hat{f}_n \succ y$; Axiom \ref{axiomtransiincomp_new} yields $\hat{f}_n \succ \hat{g}$.

 Therefore, we have proved that, for all $n \geq 1$, $\hat{f}_n \succ \hat{g}$. By Axiom \ref{axiomcont}, we conclude that $\hat{f} \succsim \hat{g}$, which, by Axiom \ref{axiomcertain}, is equivalent to $f \succsim g$.

\section{Proof of Theorem \ref{propconcor}}

We prove $(i)$; the proof of $(ii)$ is symmetric.

\medskip\noindent\textbf{Sufficiency.} Assume $\succsim$ admits the representation $(u,C,D)$ with $D \subseteq C$. Let $f \in \mathcal{F}$, $x \in X$ be such that $f \succsim x$, and let $h \in \mathcal{F}$, and $\alpha \in (0,1]$. 

Consider first the maxmin evaluation:
\begin{align*}
\min_{p\in C} \int u(\alpha f + (1-\alpha)h)dp &\geq \alpha \min_{p\in C}\int u(f)dp + (1-\alpha) \min_{p\in C}\int u(h)dp\\
&\geq \alpha u(x) + (1-\alpha) \min_{p\in C}\int u(h)dp \\
&=\min_{p\in C} \int u(\alpha x + (1-\alpha)h)dp.
\end{align*}

The second inequality follows from $f \succsim x$, which implies $\min_{p\in C}\int u(f)dp \geq u(x)$.

Consider now the maxmax evaluation. Let $q \in \argmax_{p\in D}\int u(h)dp$. The inclusion $D \subseteq C$ implies $\int u(f)dq \geq \min_{p\in C}\int u(f)dp \geq u(x)$. As a consequence,
\begin{align*}
\max_{p\in D} \int u(\alpha f + (1-\alpha)h)dp &\geq \alpha \int u(f)dq + (1-\alpha)\max_{p\in D}\int u(h)dp \\
&\geq \alpha u(x) + (1-\alpha)\max_{p\in D}\int u(h)dp\\
&=\max_{p\in D} \int u(\alpha x + (1-\alpha)h)dp.
\end{align*}

We have proved that $\alpha f + (1-\alpha)h \succsim \alpha x + (1-\alpha)h$, which is the conclusion of Axiom \ref{axiompess}.

\bigskip\noindent\textbf{Necessity.} By contraposition, assume $\succsim$ admits the representation $(u,C,D)$ with $D \not\subseteq C$, \text{i.e.} there is $p^* \in D \setminus C$. Then, there exist $f \in \mathcal{F}$ and $x \in X$ such that\footnote{We developed the argument in the proof of Theorem \ref{theo:characterization}.}
$$\int u(f)dp^* < u(x) \leq \min_{p\in C}\int u(f)dp.$$
Since $C \cap D \neq \emptyset$, one has $\max_{p \in D}\int u(f)dp \geq \min_{p\in C}\int u(f)dp \geq u(x)$. This yields $f \succsim x$. 

Now, let $h \in \mathcal{F}$ and $\alpha \in (0,1)$ be such that $\alpha f + (1-\alpha)h$ is a constant act.\footnote{Such $\alpha \in (0,1)$ and $h \in \mathcal{F}$ always exist. We prove it for completeness. It is trivially true if $f$ is a constant act. Suppose it is not, and let $\{x_1,...,x_k\}$ denote the set of different values taken by $f$ over $S$ ($k\geq 2$). Let $c=\frac{1}{k}\sum_{i=1}^k x_i \in X$. Now, define $h: S \to X$ by $h(s)=\frac{c-\alpha x_{i(s)}}{1-\alpha}$, where $i(s)$ is the index $i$ such that $x_i=f(s)$, for all $s \in S$. Then, for all $s \in S$, $h(s)=\frac{1}{k-1}\sum_{j\neq i(s)} x_j$, which is in $X$, so that $h$ is in $\mathcal{F}$. By construction, for all $s \in S$, $\frac{1}{k}f(s) +(1-\frac{1}{k})h(s)=c$.} The fact that $\alpha f + (1-\alpha)h$ is a constant act implies
$$\max_{p\in D}\int u(\alpha f + (1-\alpha)h)dp = \alpha \int u(f)dp^* + (1-\alpha)\int u(h)dp^*.$$
Since $\int u(f)dp^* < u(x)$, one has

\begin{align*}
\alpha \int u(f)dp^* + (1-\alpha)\int u(h)dp^* &< \alpha u(x) + (1-\alpha)\int u(h)dp^* \\
&\leq \max_{p\in D}\int u(\alpha x + (1-\alpha)h)dp.
\end{align*}

Hence $\alpha f + (1-\alpha)h \varnotsuccsim \alpha x + (1-\alpha)h$, that is, Axiom \ref{axiompess} does not hold. 

\section{Proof of Theorem \ref{theorem-extention}}

We prove that $(i)$ implies $(ii)$; the converse is routine. Assume $\succsim$ is a hope-and-prepare preference with unique representation $(u, C, D)$, and $\succsim^*$ is an invariant biseparable extension of $\succsim$. Since $\succsim^*$ is a complete preorder satisfying continuity, certainty independence, and monotonicity, Lemma 1 in \cite{GM2004} implies that there are a monotonic, constant-linear functional $I : B_0(\Sigma) \to \mathbb{R}$ and a non-constant affine function $u' : X \to \mathbb{R}$ such that, for all $f, g \in \mathcal{F}$,
$$f \succsim^* g \iff I(u'(f)) \geq I(u'(g)),$$
where $I$ is unique and $u'$ is unique up to positive affine transformation. Since $\succsim^*$ extends $\succsim$, the two preferences agree on constant acts, so one can assume without loss of generality $u' = u$.

Define $I'(\varphi) = \min_{p\in C}\int \varphi dp$ and $I''(\varphi) = \max_{p\in D}\int \varphi dp$ for all $\varphi \in B_0(\Sigma)$. These are monotonic and constant-linear, and $I' \leq I''$ since $C \cap D \neq \emptyset$.

We verify condition $(i)$ of Lemma A.3 in \cite{FRY2022}: for all $\varphi, \psi \in B_0(\Sigma)$, $I'(\varphi) \geq I'(\psi)$ and $I''(\varphi) \geq I''(\psi)$ implies $I(\varphi) \geq I(\psi)$. Since $u(X)$ is an interval, for any $\varphi, \psi \in B_0(\Sigma)$ there exist $a > 0$ and $b \in \mathbb{R}$ such that $a\varphi + b = u(f)$ and $a\psi + b = u(g)$ for some $f, g \in \mathcal{F}$. By constant-linearity, $I'(u(f)) \geq I'(u(g))$ and $I''(u(f)) \geq I''(u(g))$, so $f \succsim g$. Since $\succsim^*$ extends $\succsim$, $f \succsim^* g$, giving $I(u(f)) \geq I(u(g))$, and constant-linearity yields $I(\varphi) \geq I(\psi)$.

By Lemma A.3 in \cite{FRY2022}, there exists $\alpha \in [0,1]$ such that $I = \alpha I' + (1-\alpha)I''$, giving
$$I(u(f)) = \alpha \min_{p\in C}\int u(f)dp + (1-\alpha)\max_{p\in D}\int u(f)dp \: \text{ for all } f \in \mathcal{F}.$$ Finally, if $\succsim$ is incomplete, there exists $f \in \mathcal{F}$ with $\min_{p\in C}\int u(f)dp < \max_{p\in D}\int u(f)dp$. For any $\alpha' \neq \alpha$, one would have $\alpha' I'(u(f)) + (1-\alpha')I''(u(f)) \neq I(u(f))$, contradicting the uniqueness of $I$. Thus, $\alpha$ is unique.

\section{Proof of Lemma \ref{lemma:dual}}
Let $V(h) = \alpha \min_{p \in C} \int u(h)dp + (1-\alpha) \max_{p \in C} \int u(h)dp$ represent $\unrhd$. Fix $f, g \in \mathcal{F}$ with complements $\bar{f}, \bar{g}$ and a constant act $x$ such that $\frac{1}{2}f + \frac{1}{2}\bar{f} \sim x \sim \frac{1}{2}g + \frac{1}{2}\bar{g}$.

By complementarity, $\frac{1}{2}f + \frac{1}{2}\bar{f}$ has constant utility across states, equal to $u(x)$ since it is indifferent to $x$. As $u$ is affine, $u(\bar{f}(s)) = 2u(x) - u(f(s))$ for all $s$, and likewise for $\bar{g}$. Thus,
$$\min_{p \in C} \int u(\bar{f})dp = 2u(x) - \max_{p \in C} \int u(f)dp, \qquad \max_{p \in C} \int u(\bar{f})dp = 2u(x) - \min_{p \in C} \int u(f)dp,$$
so $V(\bar{f}) = 2u(x) - V_d(f)$ and $V(\bar{g}) = 2u(x) - V_d(g)$. Therefore,
$$f \unrhd_d g \iff \bar{g} \unrhd \bar{f} \iff V(\bar{g}) \geq V(\bar{f}) \iff V_d(f) \geq V_d(g).$$
The last inequality does not involve $\bar{f}, \bar{g}$, or $x$, so $V_d$ represents $\unrhd_d$. 

\section{Proof of Proposition \ref{prop:recovery}}
For $h \in \mathcal{F}$, write $a_h = \min_{p \in C} \int u(h)dp$ and $b_h = \max_{p \in C} \int u(h)dp$. Since $c_h$ and $z_h$ are the certainty equivalents of $h$ under $\unrhd$ and $\unrhd_d$, Lemma \ref{lemma:dual} gives
$$u(c_h) = \alpha a_h + (1-\alpha) b_h, \qquad u(z_h) = (1-\alpha) a_h + \alpha b_h.$$
The mixtures in the two conditions defining $\succsim^\circ$ are constant acts, so $\unrhd$ ranks them by $u$. As $u$ is affine,
$$u\big(\alpha c_f + (1-\alpha) z_g\big) - u\big(\alpha c_g + (1-\alpha) z_f\big) = \alpha[u(c_f) - u(c_g)] - (1-\alpha)[u(z_f) - u(z_g)].$$
Substituting the identities and using $\alpha^2 - (1-\alpha)^2 = 2\alpha - 1$, this difference equals $(2\alpha - 1)(a_f - a_g)$. Similarly,
$$u\big(\alpha z_f + (1-\alpha) c_g\big) - u\big(\alpha z_g + (1-\alpha) c_f\big) = \alpha[u(z_f) - u(z_g)] - (1-\alpha)[u(c_f) - u(c_g)],$$
which equals $(2\alpha - 1)(b_f - b_g)$. Since $2\alpha - 1 > 0$, the first condition holds if and only if $a_f \geq a_g$, and the second if and only if $b_f \geq b_g$. Thus, $f \succsim^\circ g$ if and only if $a_f \geq a_g$ and $b_f \geq b_g$. By Definition \ref{defunanimrep}, these inequalities characterize the concordant hope-and-prepare preference with representation $(u, C, C)$, so $\succsim^\circ$ is that preference. By Theorem \ref{theorem-extention}, its invariant biseparable completion is the standard $\alpha$-MEU preference with set of priors $C$, namely $\unrhd$. Therefore, $\succsim^\circ$ is the concordant hope-and-prepare preference whose invariant biseparable completion is $\unrhd$.

\section{Proof of Proposition \ref{prop_comp_Bew_Two}}

\noindent$(i)$ \textbf{Sufficiency.} Suppose $C_{HP} \cup D_{HP} \subseteq C_B$. If $f \succsim_B g$, then $\int u(f)dp \geq \int u(g)dp$ for all $p \in C_B$, and thus for all $p \in C_{HP} \cup D_{HP}$. It follows that
$$\min_{p \in C_{HP}} \int u(f)dp \geq \min_{p \in C_{HP}} \int u(g)dp \quad \text{and} \quad \max_{p \in D_{HP}} \int u(f)dp \geq \max_{p \in D_{HP}} \int u(g)dp,$$
so that $f \succsim_{HP} g$.

\smallskip\noindent\textbf{Necessity.} Suppose $C_{HP} \cup D_{HP} \not\subseteq C_B$. First, assume $p^* \in C_{HP} \setminus C_B$. Since $C_B$ is convex and compact, the separating hyperplane theorem implies the existence of $f \in \mathcal{F}$ and $x \in X$ satisfying
\begin{align*}
\min_{p \in C_B} \int u(f)dp > u(x) > \int u(f)dp^*.
\end{align*}
The left inequality yields $f \succsim_B x$. The right inequality yields 
$$\min_{p \in C_{HP}} \int u(f)dp \leq \int u(f)dp^* < u(x),$$ so $f \varnotsuccsim_{HP} x$, contradicting $\succsim_B \subseteq \succsim_{HP}$.

Now assume $p^* \in D_{HP} \setminus C_B$. By the same argument, there exist $f \in \mathcal{F}$ and $x \in X$ with
$$\int u(f)dp^* > u(x) > \max_{p \in C_B} \int u(f)dp.$$
The right inequality yields $x \succsim_B f$. The left inequality yields 
$$\max_{p \in D_{HP}} \int u(f)dp \geq \int u(f)dp^* > u(x),$$ so that $x \varnotsuccsim_{HP} f$, again a contradiction.

\medskip\noindent$(ii)$ \textbf{Sufficiency.} Suppose $C_{HP} \subseteq C_T$ and $D_{HP} \subseteq D_T$. If $f \succsim_T g$, then $\min_{p \in C_T} \int u(f)dp \geq \max_{p \in D_T} \int u(g)dp$. The set inclusions yield
$$\min_{p \in C_{HP}} \int u(f)dp \geq \min_{p \in C_T} \int u(f)dp \geq \max_{p \in D_T} \int u(g)dp \geq \max_{p \in D_{HP}} \int u(g)dp.$$

In addition, since $C_{HP} \cap D_{HP} \neq \emptyset$,
$$\max_{p \in D_{HP}} \int u(f)dp \geq \min_{p \in C_{HP}} \int u(f)dp \geq \max_{p \in D_{HP}} \int u(g)dp.$$
We have proved that $f \succsim_{HP} g$.

\smallskip\noindent\textbf{Necessity.} Suppose $\succsim_T \subseteq \succsim_{HP}$. If $p^* \in C_{HP} \setminus C_T$, the separating hyperplane theorem implies the existence of $f \in \mathcal{F}$ and $x \in X$ satisfying
$$\min_{p \in C_T} \int u(f)dp > u(x) > \int u(f)dp^*.$$
The left inequality yields $f \succsim_T x$, since $\min_{p \in C_T} \int u(f)dp \geq u(x) = \max_{p \in D_T} u(x)$. The right yields 
$$\min_{p \in C_{HP}} \int u(f)dp \leq \int u(f)dp^* < u(x),$$ so that $f \varnotsuccsim_{HP} x$, a contradiction.

If $p^* \in D_{HP} \setminus D_T$, the separating hyperplane theorem yields $f \in \mathcal{F}$ and $x \in X$ with
$$\int u(f)dp^* > u(x) > \max_{p \in D_T} \int u(f)dp.$$
The right inequality gives $x \succsim_T f$, since $u(x) = \min_{p \in C_T} u(x) \geq \max_{p \in D_T} \int u(f)dp$. The left gives 
$$\max_{p \in D_{HP}} \int u(f)dp \geq \int u(f)dp^* > u(x),$$ so $x \varnotsuccsim_{HP} f$, a contradiction.

\section{Proof of Theorem \ref{thm:choquet}}
\medskip\noindent\textbf{Sufficiency.} Assume there exist a non-constant affine $u:X\to\mathbb R$ and capacities $\nu_p,\nu_o$ on $(S,\Sigma)$, with $\int u\circ fd\nu_p\le\int u\circ fd\nu_o$ for all $f\in\mathcal F$, representing $\succsim$ as in Theorem \ref{thm:choquet}. The satisfaction of Axioms \ref{axiomclassic}, \ref{axiomcomono}, \ref{axiommono}, and \ref{axiomtransiincomp_new} is readily checked. We focus on Axiom \ref{axiomcontprime}. Let $f,g,f',g'\in\mathcal F$ and, for $k\in\{p,o\}$, let
\begin{equation*}
\phi_k(\lambda)=\int u\circ\big(\lambda f+(1-\lambda)g\big)d\nu_k
 -\int u\circ\big(\lambda f'+(1-\lambda)g'\big)d\nu_k.
\end{equation*}
Each $\phi_k$ is continuous on $[0,1]$, since $\lambda\mapsto u\circ(\lambda f+(1-\lambda)g)$ and $\lambda\mapsto u\circ(\lambda f' +(1-\lambda)g')$ are affine and the Choquet integral is $1$-Lipschitz in the supremum norm. By assumption, $\{\lambda\in[0,1]:\lambda a+(1-\lambda)b\succsim\lambda c+(1-\lambda)d\} =\{\lambda:\phi_p(\lambda)\ge0\}\cap\{\lambda:\phi_o(\lambda)\ge0\}$, which is closed.

\bigskip\noindent\textbf{Necessity.}
Let $\succsim$ satisfy Axioms \ref{axiomclassic}, \ref{axiomcontprime}, \ref{axiomcomono}, \ref{axiommono} and \ref{axiomtransiincomp_new}. We note that Axiom \ref{axiomcontprime} implies Axiom \ref{axiomcont}, and that Axiom \ref{axiomcomono} coincides with Axiom \ref{axiomcertain} on $X$. Thus, there is a non-constant affine function $u:X\to\mathbb R$, unique up to a positive affine transformation such that, for all $x,y \in X$,  $x\succsim y$ if and only if $u(x)\ge u(y)$. Define $\succsim_p$ and $\succsim_o$ as in the proof of Theorem \ref{theo:characterization}. Then, Steps \ref{step:wo} and Step \ref{step:connection} hold without modification. 

The proof of Step \ref{step:properties} remains valid when Axiom \ref{axiomcertain} is replaced by Axiom \ref{axiomcomono}, since the only triples for which we use Axiom \ref{axiomcertain} in the original proof are pairwise comonotonic. 

Therefore, $\succsim_p$ and $\succsim_o$ are complete preorders satisfying continuity, monotonicity, and certainty independence, and are thus invariant biseparable. By Lemma 1 in \cite{GM2004}, there are monotonic, constant-linear functionals $I_p,I_o:B_0(\Sigma)\to\mathbb R$ such that, for all $f,g\in\mathcal F$,
\begin{equation*}
f\succsim_p g\iff I_p(u\circ f)\ge I_p(u\circ g), \qquad f\succsim_o g\iff I_o(u\circ f)\ge I_o(u\circ g),
\end{equation*}
and $I_p(u(x))=I_o(u(x))=u(x)$ for all $x\in X$.

We now show that $I_p$ is comonotonically superadditive.\footnote{The definition of comonotonic acts is readily adapted to mappings in $B_0(\Sigma)$.} A symmetric argument shows that $I_o$ is comonotonically subadditive. 

Let $f,h\in\mathcal F$ be comonotonic and $\alpha\in[0,1]$. Since $I_p$ is monotonic and constant-linear, $\min_{s\in S}u(f(s))\le I_p(u\circ f)\le\max_{s\in S}u(f(s))$. Both bounds lie in $u(X)$, which is an interval, so that $I_p(u\circ f)\in u(X)$ and there is $c_f\in X$ such that  $u(c_f)=I_p(u\circ f)$, \textit{i.e.} $f\sim_p c_f$. Let $c_h$ be defined in the same way with respect to $h$. Since $c_f \in X$, the comparison $f\succsim_p c_f$ yields $f\succsim c_f$, by definition of $\succsim_p$, and, similarly, $h\succsim c_h$. The triples $(f,c_f,h)$ and $(h,c_h,c_f)$ are pairwise comonotonic: applying twice Axiom~\ref{axiomcomono} yields
\[
\alpha f+(1-\alpha)h \succsim \alpha c_f+(1-\alpha)h \succsim \alpha c_f+(1-\alpha)c_h.
\]
By Step~\ref{step:connection}, the comparisons above imply $\alpha f+(1-\alpha)h\succsim_p\alpha c_f+(1-\alpha)c_h$, which is equivalent to
\begin{equation}\label{eq:mixsuper}
 I_p\big(u\circ(\alpha f+(1-\alpha)h)\big)\ge\alpha I_p(u\circ f)+(1-\alpha) I_p(u\circ h).
\end{equation}
We extend \eqref{eq:mixsuper} to arbitrary comonotonic $\varphi,\psi\in B_0(\Sigma)$. As $\varphi,\psi$ are bounded and $u(X)$ is a non-degenerate interval, there are $a>0$ and $b\in\mathbb R$ such that $a\varphi+b$ and $a\psi+b$ take value in $u(X)$. Let $f,g \in \mathcal{F}$ be such that $a\varphi+b=u\circ f$
and $a\psi+b=u\circ h$. The acts $f$ and $h$ are comonotonic. Since $I_p$ is positively homogeneous, taking  $\alpha=\tfrac12$ in \eqref{eq:mixsuper} gives $I_p(u\circ f+u\circ h)\ge I_p(u\circ f)+I_p(u\circ h)$. Since $I_p$ is constant-linear, $I_p(u\circ f+u\circ h)=I_p\big(a(\varphi+\psi)+2b\big)=a I_p(\varphi+\psi)+2b$
and $I_p(u\circ f)+I_p(u\circ h)=a\big(I_p(\varphi)+I_p(\psi)\big)+2b$. Dividing by $a>0$, we have proved that 
\[
 I_p(\varphi+\psi)\ge I_p(\varphi)+I_p(\psi).
\]

We will prove that the inequality above is always an equality, meaning that $I_p$ is comonotonically additive; the argument for $I_o$ is symmetric. By contradiction, assume that there are comonotonic $\phi$ and $\psi$ in $B_0(\Sigma)$ such that the inequality above is strict. Then, by the properties of $I_p$ and $u$, there must exist comonotonic acts $f,g \in \mathcal{F}$ such that 

\[
 I_p \left(u\circ\left(\tfrac12 f+\tfrac12 g \right)\right)>\tfrac12 I_p(u\circ f)+\tfrac12 I_p(u\circ g).
\]

In addition, given positive homogeneity and constant-linearity, we may assume without loss of generality that $I_p(u\circ f)=I_p(u\circ g)= \mu$, and that all the mappings introduced below also take value in
$u(X)$. 

Assume $I_o(u\circ f)\ge I_o(u\circ g)$, and set
$\rho:=I_p \left(u\circ \left(\tfrac12 f+\tfrac12 g\right)\right)-\mu>0$. The argument in the case $I_o(u\circ f)\le I_o(u\circ g)$ is the same, simply reversing the role of $f$ and $g$.

Let $f'\in\mathcal F$ such that $u\circ f'=u\circ f+\rho$. Then $f'$ is comonotonic with $g$, and
\[
I_p(u\circ f')=\mu+\rho>\mu=I_p(u\circ g),\qquad I_o(u\circ f')=I_o(u\circ f)+\rho\ge I_o(u\circ g)+\rho>I_o(u\circ g).
\]
This implies $f'\succ_p g$ and $f'\succ_o g$, and thus $f'\succ g$. Since the triple $(f',g,f')$ is pairwise comonotonic, Axiom \ref{axiomcomono} implies $f'\succsim\tfrac12 g+\tfrac12 f'$. Then, by Step \ref{step:connection}, we get $f'\succsim_p\tfrac12 g+\tfrac12 f'$, that is, $I_p(u\circ f')\ge I_p \left(u\circ \left(\tfrac12 g+\tfrac12 f' \right)\right)$. Yet, 

\[
I_p \left(u\circ \left(\tfrac12 g+\tfrac12 f'\right)\right) =I_p \left(u\circ \left(\tfrac12 g+\tfrac12 f \right)\right)+\tfrac{\rho}{2} =(\mu+\rho)+\tfrac{\rho}{2}>\mu+\rho=I_p(u\circ f'),
\]
a contradiction. Thus, $I_p$ is comonotonically additive.

We have thus proved that  $\succsim_p$ and $\succsim_o$ satisfy comonotonic independence. Each is a complete preorder satisfying continuity, monotonicity, and comonotonic independence, so by the representation theorem of \cite{S1989}, there are unique capacities $\nu_p,\nu_o$ on $(S,\Sigma)$ such that, for all $f \in \mathcal{F}$, $I_p(u \circ f)=\int u \circ f \: d\nu_p$ and $I_o(u \circ f)=\int u \circ f \: d\nu_o$. 

Next we show that for all $f \in \mathcal{F}$, $\int u\circ f d\nu_p\le\int u\circ f d\nu_o$. Suppose not. Since $u(X)$ is an interval, there is $x\in X$ with $\int u\circ f d\nu_p>u(x)>\int u\circ f d\nu_o$. The left inequality gives $f\succ_p x$, hence $f\succsim x$ by definition of $\succsim_p$, hence $f\succsim_o x$ by Step \ref{step:connection}. The right inequality gives $x\succ_o f$, a contradiction.

It remains to prove that $[\int u\circ f d\nu_p\ge\int u\circ g d\nu_p \text{ and}
\int u\circ f d\nu_o\ge\int u\circ g d\nu_o]$ implies $f\succsim g$. Let $x^*,x_*\in X$ be such that $u(x^*)>u(x_*)$ ($u$ is non constant). For each integer $n\ge1$, let $f_n=\left(1-\tfrac1n \right)f+\tfrac1n x^*$ and $g_n=\left(1-\tfrac1n \right)g+\tfrac1n x_*$. Then, for $k\in\{p,o\}$, for all $n \geq 1$,
\[
\int u\circ f_n d\nu_k-\int u\circ g_n d\nu_k =\left(1-\tfrac1n\right) \left(\int u\circ f d\nu_k-\int u\circ g d\nu_k\right) +\tfrac1n\big(u(x^*)-u(x_*)\big)>0 .
\]
Hence $f_n\succ_p g_n$ and $f_n\succ_o g_n$. A similar argument as in the Conclusion of the Proof of Theorem \ref{theo:characterization}, yields $f_n\succ g_n$. Letting $\lambda_n=1-\tfrac1n$, one has $\lambda_n f+(1-\lambda_n)x^*\succsim\lambda_n g+(1-\lambda_n)x_*$ for all $n$. The set
$\{\lambda\in[0,1]:\lambda f+(1-\lambda)x^*\succsim\lambda g+(1-\lambda)x_*\}$ is closed by Axiom~\ref{axiomcontprime} and contains the sequence $(\lambda_n)_{n \geq 1}$, which converges to 1. Therefore, $f\succsim g$.

\singlespacing
\small
\bibliographystyle{ecta-fullname}
\bibliography{bibdecision}

@article{dur2018identifying,
  title={Identifying the harm of manipulable school-choice mechanisms},
  author={Dur, Umut and Hammond, Robert G and Morrill, Thayer},
  journal={American Economic Journal: Economic Policy},
  volume={10},
  number={1},
  pages={187--213},
  year={2018},
  publisher={American Economic Association}
}

@article{bardier2024hoping,
  title={Hoping for the best while preparing for the worst in the face of uncertainty: a new type of incomplete preferences},
  author={Bardier, Pierre and Dong-Xuan, Bach and Nguyen, Van-Quy},
  journal={Mimeo},
  year={2025}
}

@article{morewedge2018betting,
  title={Betting your favorite to win: Costly reluctance to hedge desired outcomes},
  author={Morewedge, Carey K and Tang, Simone and Larrick, Richard P},
  journal={Management Science},
  volume={64},
  number={3},
  pages={997--1014},
  year={2018},
  publisher={INFORMS}
}

@article{donkor2023identity,
  title={Identity and economic incentives},
  author={Donkor, Kwabena and Goette, Lorenz and M{\"u}ller, Maximilian W and Dimant, Eugen and Kurschilgen, Michael},
  year={2023},
  publisher={CESifo Working Paper}
}

@article{kossuth2020does,
  title={Does it pay to bet on your favourite to win? Evidence on experienced utility from the 2018 FIFA World Cup experiment},
  author={Kossuth, Lajos and Powdthavee, Nattavudh and Harris, Donna and Chater, Nick},
  journal={Journal of Economic Behavior \& Organization},
  volume={171},
  pages={35--58},
  year={2020},
  publisher={Elsevier}
}

@article{adam2024event,
  title={Event Valence and Subjective Probability},
  author={Adam, Brandenburger and Ghirardato, Paolo and Pennesi, Daniele and Stanca, Lorenzo Maria and others},
  journal={Carlo Alberto Notebooks},
  volume={717},
  year={2024}
}

@article{schroder2011investment,
  title={Investment under ambiguity with the best and worst in mind},
  author={Schr{\"o}der, David},
  journal={Mathematics and Financial Economics},
  volume={4},
  pages={107--133},
  year={2011},
  publisher={Springer}
}

@article{taylor2017framework,
  title={A framework to improve surgeon communication in high-stakes surgical decisions: best case/worst case},
  author={Taylor, Lauren J and Nabozny, Michael J and Steffens, Nicole M and Tucholka, Jennifer L and Brasel, Karen J and Johnson, Sara K and Zelenski, Amy and Rathouz, Paul J and Zhao, Qianqian and Kwekkeboom, Kristine L and others},
  journal={JAMA surgery},
  volume={152},
  number={6},
  pages={531--538},
  year={2017},
  publisher={American Medical Association}
}

@article{pathak2008leveling,
  title={Leveling the playing field: Sincere and sophisticated players in the Boston mechanism},
  author={Pathak, Parag A and S{\"o}nmez, Tayfun},
  journal={American Economic Review},
  volume={98},
  number={4},
  pages={1636--1652},
  year={2008},
  publisher={American Economic Association}
}

@article{nascimento2011,
  title={A class of incomplete and ambiguity averse preferences},
  author={Nascimento, Leandro and Riella, Gil},
  journal={Journal of Economic Theory},
  volume={146},
  number={2},
  pages={728--750},
  year={2011},
  publisher={Elsevier}
}

@article{chateauneuf2024alpha,
  title={Alpha-maxmin as an aggregation of two selves},
  author={Chateauneuf, Alain and Faro, Jos{\'e} Heleno and Tallon, Jean-Marc and Vergopoulos, Vassili},
  journal={Journal of Mathematical Economics},
  pages={103006},
  year={2024},
  publisher={Elsevier}
}

@techreport{cusumano2021,
  title={Twofold preferences under uncertainty},
  author={Cusumano, Carlo M and Miyashita, Masaki},
  year={2021},
  institution={Working paper, Yale}
}

@article{GS1989,
title = {\href{http://www.sciencedirect.com/science/article/B6VBY-4582D7Y-1W/2/eba21cf3548fc796bbf94d33ae7d45c7}{Maxmin expected utility with non-unique prior}},
journal = "Journal of Mathematical Economics",
volume = "18",
number = "2",
pages = "141 - 153",
year = "1989",
author = "Itzhak Gilboa and David Schmeidler"
}

@article{S1989,
     title = {\href{http://www.jstor.org/stable/1911053}{Subjective Probability and Expected Utility without Additivity}},
     author = {Schmeidler, David},
     journal = {Econometrica},
     volume = {57},
     number = {3},
     pages = {571--587},
     year = {1989},
     }

@article{B2002,
title = {\href{http://www.springerlink.com/content/aed5mt60tbe7eu0u/}{Knightian decision theory. Part I}},
journal = "Decisions in Economics and Finance",
volume = "25",
number = "2",
pages = "79-110",
year = "2002",
author = "Bewley, Truman F."
}

@article{GM2002,
title = {\href{https://www.sciencedirect.com/science/article/pii/S0022053101928157}{Ambiguity Made Precise: A Comparative Foundation}},
journal = {Journal of Economic Theory},
volume = {102},
number = {2},
pages = {251-289},
year = {2002},
issn = {0022-0531},
doi = {https://doi.org/10.1006/jeth.2001.2815},
author = {Paolo Ghirardato and Massimo Marinacci}
}

@article{GM2004,
title = {\href{https://www.sciencedirect.com/science/article/pii/S0022053104000262}{Differentiating ambiguity and ambiguity attitude}},
journal = {Journal of Economic Theory},
volume = {118},
number = {2},
pages = {133-173},
year = {2004},
issn = {0022-0531},
doi = {https://doi.org/10.1016/j.jet.2003.12.004},
author = {Paolo Ghirardato and Fabio Maccheroni and Massimo Marinacci}
}

@article{TS1981,
 ISSN = {00223808, 1537534X},
 journal = {Journal of Political Economy},
 number = {2},
 pages = {392--406},
 publisher = {University of Chicago Press},
 title = {\href{https://www.journals.uchicago.edu/doi/10.1086/260971}{An Economic Theory of Self-Control}},
 urldate = {2023-04-24},
 volume = {89},
 author = {Richard H. Thaler and H. M. Shefrin},
 year = {1981}
}

@inbook{TK2015,
author = {Trautmann, Stefan T. and van de Kuilen, Gijs},
publisher = {John Wiley \& Sons, Ltd},
isbn = {9781118468333},
title = {Ambiguity Attitudes},
booktitle = {\href{https://onlinelibrary.wiley.com/doi/abs/10.1002/9781118468333.ch3}{The Wiley Blackwell Handbook of Judgment and Decision Making}},
chapter = {3},
pages = {89-116},
doi = {https://doi.org/10.1002/9781118468333.ch3},
url = {},
eprint = {https://onlinelibrary.wiley.com/doi/pdf/10.1002/9781118468333.ch3},
year = {2015},
}

@article{knight1921risk,
  title={Risk, Uncertainty and Profit},
  author={Knight, Frank H},
  journal={University of Illinois at Urbana-Champaign's Academy for Entrepreneurial Leadership Historical Research Reference in Entrepreneurship},
  year={1921}
}

@article{TM2020,
title = {\href{https://www.sciencedirect.com/science/article/pii/S002205311830629X}{Obvious manipulations}},
journal = {Journal of Economic Theory},
volume = {185},
pages = {104970},
year = {2020},
issn = {0022-0531},
doi = {https://doi.org/10.1016/j.jet.2019.104970},
url = {},
author = {Peter Troyan and Thayer Morrill}
}

@article{E2022,
title = {\href{https://www.sciencedirect.com/science/article/pii/S0022053122000382}{Twofold multiprior preferences and failures of contingent reasoning}},
journal = {Journal of Economic Theory},
volume = {202},
pages = {105448},
year = {2022},
issn = {0022-0531},
doi = {https://doi.org/10.1016/j.jet.2022.105448},
url = {},
author = {Federico Echenique and Masaki Miyashita and Yuta Nakamura and Luciano Pomatto and Jamie Vinson}
}

@article{CE2022,
author = {Chandrasekher, Madhav and Frick, Mira and Iijima, Ryota and Le Yaouanq, Yves},
title = {\href{https://onlinelibrary.wiley.com/doi/abs/10.3982/ECTA17502}{Dual-Self Representations of Ambiguity Preferences}},
journal = {Econometrica},
volume = {90},
number = {3},
pages = {1029-1061},
doi = {https://doi.org/10.3982/ECTA17502},
url = {},
eprint = {https://onlinelibrary.wiley.com/doi/pdf/10.3982/ECTA17502},
year = {2022}
}

@article{FRY2022,
title = {\href{https://www.sciencedirect.com/science/article/pii/S0022053121002118}{Objective rationality foundations for (dynamic) $\alpha$-MEU}},
journal = {Journal of Economic Theory},
volume = {200},
pages = {105394},
year = {2022},
issn = {0022-0531},
doi = {https://doi.org/10.1016/j.jet.2021.105394},
url = {},
author = {Mira Frick and Ryota Iijima and Yves {Le Yaouanq}}
}

@article{GMMS2010,
 ISSN = {00129682, 14680262},
 URL = {http://www.jstor.org/stable/40664491},
 author = {Itzhak Gilboa and Fabio Maccheroni and Massimo Marinacci and David Schmeidler},
 journal = {Econometrica},
 number = {2},
 pages = {755--770},
 publisher = {[Wiley, Econometric Society]},
 title = {\href{https://onlinelibrary.wiley.com/doi/abs/10.3982/ECTA8223}{Objective and Subjective Rationality
in a Multiple Prior Model}},
 urldate = {2023-05-04},
 volume = {78},
 year = {2010}
}

@article{AA1963,
 ISSN = {00034851},
 URL = {http://www.jstor.org/stable/2991295},
 author = {F. J. Anscombe and R. J. Aumann},
 journal = {The Annals of Mathematical Statistics},
 number = {1},
 pages = {199--205},
 publisher = {Institute of Mathematical Statistics},
 title = {\href{http://www.jstor.org/stable/2991295}{A Definition of Subjective Probability}},
 urldate = {2023-06-03},
 volume = {34},
 year = {1963}
}

@article{S2009,
  title={\href{https://onlinelibrary.wiley.com/doi/abs/10.3982/ECTA7564}{Vector expected utility and attitudes toward variation}},
  author={Siniscalchi, Marciano},
  journal={Econometrica},
  volume={77},
  number={3},
  pages={801--855},
  year={2009},
  publisher={Wiley Online Library}
}

@Inbook{GM2016,
author="Gilboa, Itzhak
and Marinacci, Massimo",
editor="Arl{\'o}-Costa, Horacio
and Hendricks, Vincent F.
and van Benthem, Johan",
title={\href{https://doi.org/10.1007/978-3-319-20451-2_21}{Ambiguity and the Bayesian Paradigm}},
bookTitle="Readings in Formal Epistemology: Sourcebook",
year="2016",
publisher="Springer International Publishing",
address="Cham",
pages="385--439",
isbn="978-3-319-20451-2",
doi="10.1007/978-3-319-20451-2_21",
}

@article{EMT2012,
author = {Etner, Johanna and Jeleva, Meglena and Tallon, Jean-Marc},
title = {\href{https://onlinelibrary.wiley.com/doi/pdf/10.1111/j.1467-6419.2010.00641.x}{Decision theory under ambiguity}},
journal = {Journal of Economic Surveys},
volume = {26},
number = {2},
pages = {234--270},
doi = {https://doi.org/10.1111/j.1467-6419.2010.00641.x},
url = {https://onlinelibrary.wiley.com/doi/abs/10.1111/j.1467-6419.2010.00641.x},
eprint = {https://onlinelibrary.wiley.com/doi/pdf/10.1111/j.1467-6419.2010.00641.x},
year = {2012}
}

@article{Xia2020,
      title={\href{https://arxiv.org/abs/2012.07509}{Decision Making under Uncertainty: A Game of Two Selves}}, 
      author={Jianming Xia},
      year={2020},
      institution={Working paper},
      eprint={2012.07509},
      archivePrefix={arXiv},
      primaryClass={econ.TH}
}

@article{Li2017,
author = {Shengwu Li},
year = {2017},
month = { },
pages = {3257-3287},
title =  {\href{https://www.aeaweb.org/articles?id=10.1257/aer.20160425}{Obviously strategy-proof mechanisms}},
volume = {107(11)},
journal = {American Economic Review},
doi = {10.1257/aer.20160425}
}

@article{Au1962,
 ISSN = {00129682, 14680262},
 author = {Robert J. Aumann},
 journal = {Econometrica},
 number = {3},
 pages = {445--462},
 publisher = {[Wiley, Econometric Society]},
 title = {\href{https://www.jstor.org/stable/1909888?origin=crossref}{Utility Theory without the Completeness Axiom}},
 volume = {30},
 year = {1962}
}

@article{DFO2004,
title = {\href{https://www.sciencedirect.com/science/article/pii/S0022053103001662}{Expected utility theory without the completeness axiom}},
journal = {Journal of Economic Theory},
volume = {115},
number = {1},
pages = {118-133},
year = {2004},
author = {Juan Dubra and Fabio Maccheroni and Efe Ok},
}

@article{GK2012,
title = {\href{https://www.sciencedirect.com/science/article/pii/S0165489612000388}{Expected multi-utility representations}},
journal = {Mathematical Social Sciences},
volume = {64},
number = {3},
pages = {242-246},
year = {2012},
author = {Tsogbadral Galaabaatar and Edi Karni}
}

@article{OOR2012,
author = {Efe, Ok and Ortoleva, Pietro and Riella, Gil},
title = {\href{https://onlinelibrary.wiley.com/doi/10.3982/ECTA8040}{Incomplete Preferences Under Uncertainty: Indecisiveness in Beliefs versus Tastes}},
journal = {Econometrica},
volume = {80},
number = {4},
pages = {1791-1808},
year = {2012}
}

@article{Hill2016,
title = {\href{https://www.sciencedirect.com/science/article/pii/S0304406816300301}{Incomplete preferences and confidence}},
journal = {Journal of Mathematical Economics},
volume = {65},
pages = {83-103},
year = {2016},
issn = {0304-4068},
author = {Brian Hill},
}

@article{MS2015,
title = {\href{https://www.sciencedirect.com/science/article/pii/S0022053114001689}{Preferences with grades of indecisiveness}},
journal = {Journal of Economic Theory},
volume = {155},
pages = {300-331},
year = {2015},
issn = {0022-0531},
doi = {https://doi.org/10.1016/j.jet.2014.11.009},
url = {},
author = {Stefania Minardi and Andrei Savochkin}
}

@article{FARO2015,
title = {\href{https://www.sciencedirect.com/science/article/pii/S0022053115000356}{Variational Bewley preferences}},
journal = {Journal of Economic Theory},
volume = {157},
pages = {699-729},
year = {2015},
issn = {0022-0531},
doi = {https://doi.org/10.1016/j.jet.2015.02.002},
url = {},
author = {José Faro}
}

@article{KARNI2020,
title = {\href{https://www.sciencedirect.com/science/article/pii/S030440682030045}{Probabilistic sophistication without completeness}},
journal = {Journal of Mathematical Economics},
volume = {89},
pages = {8-13},
year = {2020},
issn = {0304-4068},
doi = {https://doi.org/10.1016/j.jmateco.2020.04.002},
url = {8},
author = {Edi Karni}
}

@article{EICHBERGER20111684,
title = {The $\alpha$-MEU model: A comment},
journal = {Journal of Economic Theory},
volume = {146},
number = {4},
pages = {1684-1698},
year = {2011},
issn = {0022-0531},
doi = {https://doi.org/10.1016/j.jet.2011.03.019},
url = {https://www.sciencedirect.com/science/article/pii/S0022053111000512},
author = {Jürgen Eichberger and Simon Grant and David Kelsey and Gleb A. Koshevoy}
}

@article{CHATEAUNEUF2007538,
title = {Choice under uncertainty with the best and worst in mind: Neo-additive capacities},
journal = {Journal of Economic Theory},
volume = {137},
number = {1},
pages = {538-567},
year = {2007},
issn = {0022-0531},
doi = {https://doi.org/10.1016/j.jet.2007.01.017},
url = {https://www.sciencedirect.com/science/article/pii/S0022053107000300},
author = {Alain Chateauneuf and Jürgen Eichberger and Simon Grant}
}

@techreport{kopylov2002alpha,
  title={$\alpha$-Maxmin expected utility},
  author={Kopylov, I},
  year={2002},
  institution={mimeo, University of Rochester}
}

@article{GUL2015465,
title = {Hurwicz expected utility and subjective sources},
journal = {Journal of Economic Theory},
volume = {159},
pages = {465-488},
year = {2015},
issn = {0022-0531},
doi = {https://doi.org/10.1016/j.jet.2015.05.007},
url = {https://www.sciencedirect.com/science/article/pii/S0022053115000861},
author = {Faruk Gul and Wolfgang Pesendorfer}
}

@article{klibanoff2022foundations,
  title={Foundations of ambiguity models under symmetry: $\alpha$-MEU and smooth ambiguity},
  author={Klibanoff, Peter and Mukerji, Sujoy and Seo, Kyoungwon and Stanca, Lorenzo},
  journal={Journal of Economic Theory},
  volume={199},
  pages={105202},
  year={2022},
  publisher={Elsevier}
}

@article{hill2023beyond,
  title={Beyond uncertainty aversion},
  author={Hill, Brian},
  journal={Games and Economic Behavior},
  volume={141},
  pages={196--222},
  year={2023},
  publisher={Elsevier}
}

@article{HARTMANN2023105719,
title = {Strength of preference over complementary pairs axiomatizes alpha-MEU preferences},
journal = {Journal of Economic Theory},
volume = {213},
pages = {105719},
year = {2023},
issn = {0022-0531},
doi = {https://doi.org/10.1016/j.jet.2023.105719},
url = {https://www.sciencedirect.com/science/article/pii/S0022053123001151},
author = {Lorenz Hartmann}
}

@article{BGGZ2010,
    author = {Bossaerts, Peter and Ghirardato, Paolo and Guarnaschelli, Serena and Zame, William R.},
    title = "{Ambiguity in Asset Markets: Theory and Experiment}",
    journal = {The Review of Financial Studies},
    volume = {23},
    number = {4},
    pages = {1325-1359},
    year = {2010},
    month = {01},
    issn = {0893-9454},
    doi = {10.1093/rfs/hhp106},
    url = {https://doi.org/10.1093/rfs/hhp106},
    eprint = {https://academic.oup.com/rfs/article-pdf/23/4/1325/24430084/hhp106.pdf},
}

@article{https://doi.org/10.3982/QE243,
author = {Ahn, David and Choi, Syngjoo and Gale, Douglas and Kariv, Shachar},
title = {Estimating ambiguity aversion in a portfolio choice experiment},
journal = {Quantitative Economics},
volume = {5},
number = {2},
pages = {195-223},
doi = {https://doi.org/10.3982/QE243},
url = {https://onlinelibrary.wiley.com/doi/abs/10.3982/QE243},
eprint = {https://onlinelibrary.wiley.com/doi/pdf/10.3982/QE243},
year = {2014}
}

@article{BENABOU2002419,
title = {Dynamic inconsistency and self-control: a planner–doer interpretation},
journal = {Economics Letters},
volume = {77},
number = {3},
pages = {419-424},
year = {2002},
issn = {0165-1765},
doi = {https://doi.org/10.1016/S0165-1765(02)00158-1},
url = {https://www.sciencedirect.com/science/article/pii/S0165176502001581},
author = {Roland Bénabou and Marek Pycia}
}

@article{10.1257/aer.96.5.1449,
Author = {Fudenberg, Drew and Levine, David K.},
Title = {A Dual-Self Model of Impulse Control},
Journal = {American Economic Review},
Volume = {96},
Number = {5},
Year = {2006},
Month = {December},
Pages = {1449–1476},
DOI = {10.1257/aer.96.5.1449},
URL = {https://www.aeaweb.org/articles?id=10.1257/aer.96.5.1449}}

@article{10.1257/aer.98.4.1312,
Author = {Brocas, Isabelle and Carrillo, Juan D.},
Title = {The Brain as a Hierarchical Organization},
Journal = {American Economic Review},
Volume = {98},
Number = {4},
Year = {2008},
Month = {September},
Pages = {1312–46},
DOI = {10.1257/aer.98.4.1312},
URL = {https://www.aeaweb.org/articles?id=10.1257/aer.98.4.1312}}

@article{marinacci2002probabilistic,
  title={Probabilistic sophistication and multiple priors},
  author={Marinacci, Massimo},
  journal={Econometrica},
  volume={70},
  number={2},
  pages={755--764},
  year={2002},
  publisher={JSTOR}
}

@article{doi:10.7326/0003-4819-138-5-200303040-00028,
author={Back, Anthony L and Arnold, Robert M and Quill, Timothy E},
title = {Hope for the Best, and Prepare for the Worst},
journal = {Annals of Internal Medicine},
volume = {138},
number = {5},
pages = {439-443},
year = {2003},
doi = {10.7326/0003-4819-138-5-200303040-00028},
note ={PMID: 12614110},
URL = {https://www.acpjournals.org/doi/abs/10.7326/0003-4819-138-5-200303040-00028},
eprint ={https://www.acpjournals.org/doi/pdf/10.7326/0003-4819-138-5-200303040-00028}
}

@article{CETTOLIN2019547,
title = {Revealed preferences under uncertainty: Incomplete preferences and preferences for randomization},
journal = {Journal of Economic Theory},
volume = {181},
pages = {547-585},
year = {2019},
issn = {0022-0531},
doi = {https://doi.org/10.1016/j.jet.2019.03.002},
url = {https://www.sciencedirect.com/science/article/pii/S0022053119300262},
author = {Elena Cettolin and Arno Riedl}
}

@techreport{nielsen2022revealed,
  title={Revealed incomplete preferences},
  author={Nielsen, Kirby and Rigotti, Luca},
  year={2024},
  institution={Working Paper}
}

@book{binmore,
 ISBN = {9780691149899},
 URL = {http://www.jstor.org/stable/j.ctt7szmq},
 author = {Ken Binmore},
 edition = {STU - Student edition},
 publisher = {Princeton University Press},
 title = {Rational Decisions},
 urldate = {2024-12-20},
 year = {2009}
}

@techreport{grant2020worst,
  title={Worst-and best-case expected utility and ordinal meta-utility},
  author={Grant, Simon and Rich, Patricia and Stecher, Jack},
  year={2020},
  institution={Working Paper}
}

@techreport{hurwicz1951optimality,
  title={Optimality criteria for decision making under ignorance},
  author={Hurwicz, Leonid},
  year={1951},
  institution={Cowles Commission discussion paper, statistics}
}

@article{dimmock2016ambiguity,
  title={Ambiguity attitudes in a large representative sample},
  author={Dimmock, Stephen G and Kouwenberg, Roy and Wakker, Peter P},
  journal={Management Science},
  volume={62},
  number={5},
  pages={1363--1380},
  year={2016},
  publisher={INFORMS}
}

@article{watanabe2024ambiguity,
  title={Ambiguity attitudes toward natural and artificial sources in gain and loss domains},
  author={Watanabe, Masahide and Fujimi, Toshio},
  journal={Journal of Risk and Uncertainty},
  volume={68},
  number={1},
  pages={51--75},
  year={2024},
  publisher={Springer}
}

@article{baillon2025source,
  title={Source theory: A tractable and positive ambiguity theory},
  author={Baillon, Aur{\'e}lien and Bleichrodt, Han and Li, Chen and Wakker, Peter P},
  journal={Management Science},
  year={2025},
  publisher={INFORMS}
}

\end{document}